\documentclass[envcountsame]{llncs}

\usepackage{paralist}
\usepackage[utf8]{inputenc}
\usepackage[T1]{fontenc}
\usepackage[numbers]{natbib}
\usepackage{amsfonts}
\usepackage{mathtools}
\usepackage{latexsym, amsmath, amssymb}
\usepackage{stmaryrd}
\usepackage{wasysym}
\usepackage{rotating}
\usepackage{pifont}
\usepackage{booktabs}
\usepackage{array}
\usepackage{multirow}
\usepackage{tikz}
\usepackage{color}
\usepackage{xcolor}
\usepackage{enumitem}
\usepackage{xspace}
\usepackage{thm-restate}
\usepackage{bm}

\spnewtheorem{definition}{Definition}{\bfseries}{\rmfamily}

\usepackage{xparse}
\usepackage{misc}

\DeclareDocumentCommand\mmsc{mmo}{%
    \IfNoValueT{#3}{\ensuremath{{\sf mmsc}\left({#1}, {#2}\right)}\xspace}%
    \IfNoValueF{#3}{\ensuremath{{\sf mmsc}\left({#1}, {#2}, {#3}\right)}\xspace}%
}
\newcommand{\dgpath}[1]{{\mathbf{#1}}}

\newcommand{\dgraph}{{\mathcal{G}}}
\newcommand{\C}{{\mathbf{C}}}

\newcommand{\mvf}{\ensuremath{{\sf mvf}}}
\newcommand{\mmvf}{\ensuremath{{\sf mmvf}}}
\newcommand{\reach}{\ensuremath{{\sf reach}}}
\newcommand{\weight}{\ensuremath{{\sf weight}}}
\newcommand{\SCC}{\ensuremath{{\sf SCC}}}
\newcommand{\scc}{\ensuremath{{\sf scc}}}
\newcommand{\lend}{{\ensuremath{{\sf v_{num}}}}}

\newcommand{\app}{approx}
\newcommand{\base}{\mathcal{B}}
\newcommand{\powerset}[1]{{\mathcal{P}({#1})}}
\newcommand{\newdsym}{d}
\newcommand{\newd}[2]{{\newdsym_{#2}\left(#1\right)}}

\newcommand{\assertion}[1]{\ensuremath{{\sf #1}}\xspace}

\newcommand{\J}{\mathcal{J}}

\usepackage[switch]{lineno} 
\usepackage[ruled,vlined,noend]{algorithm2e}
\newif\ifappendix
\appendixtrue 

\newcommand{\ourEL}{\ensuremath{{\cal E\!L}^\bot}\xspace}
\title{Mining \ourEL Bases with Adaptable Role Depth}
\author {
Ricardo Guimarães, 
Ana Ozaki, 
Cosimo Persia, 
 Baris Sertkaya 
}
\institute {
    Department of Informatics, University of Bergen, Norway\\
    Frankfurt University of Applied Sciences, Germany \\
    Ricardo.Guimaraes@uib.no, Ana.Ozaki@uib.no, Cosimo.Persia@uib.no, Sertkaya@fb2.fra-uas.de
}
\begin{document}
\maketitle

\begin{abstract}
In Formal Concept Analysis, a base for a finite structure is a 
set of implications that characterizes all valid implications of
the structure. This notion can be adapted to the context of Description Logic,
where the base consists of a set of concept inclusions instead of implications.
In this setting, concept expressions can be arbitrarily large.  Thus, it is not
clear whether a finite base exists and, if so, how large concept expressions may need to be.  
We first revisit results in the literature for mining \ourEL bases from finite interpretations. 
Those  mainly focus on finding a finite  base or on fixing the role depth but potentially losing
some of the valid concept inclusions with higher role depth. 
We then present a new strategy for mining \ourEL bases which
is adaptable  in the sense that it can bound the role depth of concepts
depending on the local structure of the interpretation. Our strategy
  guarantees to capture \emph{all}
\ourEL concept inclusions holding in the
interpretation,
  not only the ones up to a fixed role depth.
\end{abstract}

\section{Introduction}
Among its many applications in artificial intelligence, logic is used to formally represent knowledge.
Such knowledge, often in the form of facts and rules, enables machines to process 
complex relational data, deduce new knowledge from it,
and extract 
hidden 
relationships  in a specific domain.
A well-studied  formalism for  knowledge representation
is given by a family
of logics 
known as description logics (DLs)~\cite{DBLP:books/daglib/0041477}.
DL is the logical formalism behind 
the design of many knowledge-based applications.
However, it is often difficult and time-consuming to manually model 
in a formal language  rules and constraints that hold in a  domain of knowledge.

In this work, we consider an automatic method 
to extract 
 rules (concept inclusions (CIs)) formulated in DL from data.
This data can be, 
for instance, a collection of facts 
in a database or 
a 
knowledge graph. 
For instance, in the DBpedia knowledge graph~\cite{DBLP:journals/semweb/LehmannIJJKMHMK15},  one can represent
the 
relationship  between a city 
`${\sf a }$'
and the region 
`${\sf b }$'  
it belongs to
with the facts
\assertion{city(a)},
\assertion{region(b)},
\assertion{partof(a,b)}, and
\assertion{capital(b,a)}.
From this data, one can mine a CI 
expressing that a capital is a city that is part of a region.

To mine CIs that hold in a dataset, 
we combine notions of 
Formal Concept Analysis (FCA) \cite{ganter1999formal}
and DLs. 
FCA is a subfield of lattice theory
that provides methods for analysing datasets and identifying the dependencies
in them.
In FCA a dataset, 
also called a \emph{formal context}, 
is a table showing which objects have which attributes. 
Given a formal context, FCA methods are used to extract the dependencies
between the attributes, also called implications~(Figure~\ref{f:formal_context}).
%
A \emph{base}
is a set of implications that 
entails every valid implication of the dataset and only those 
(soundness and completeness).
It can be used for detecting erroneous or missing items in the dataset~\cite{BGSS07}.
In the DL setting, a base is a set of  CIs (an ontology)
which can serve as a starting point for ontology engineers to build 
an ontology in a domain of interest.

However, 
for some DLs and datasets, it may happen that no finite base exists.
Cyclic relationships are common in knowledge graphs and
they are the
main 
challenge for finding a finite DL base.
With only one cyclic relationship, we already have that infinitely many concepts hold in the dataset.
%
Strategies for limiting the size of concepts in the presence of cyclic dependencies
have already been investigated in the literature. 
Baader and Distel (2008) and Distel (2011)
propose
a way of mining  DL finite bases 
expressible in   
the DL $\EL^\bot_{gfp}$ 
which is the addition of
greatest fix-point semantics 
to the DL language $ \EL^\bot $~\cite{BaDi08,Distel2011}.
The semantics offered by $ \EL^\bot_{gfp} $
elegantly solves the difficulty of mining CIs from
cyclic relationships in the data. 
However, this semantics comes with two drawbacks.
Firstly, 
$ \EL^\bot_{gfp} $ concepts may be difficult to understand,
and learned CIs may be too complex to validate by
domain experts.
Secondly, there is no efficient implementation of a reasoner for
$ \EL^\bot_{gfp} $, 
even though the reasoning complexity is tractable, like for $ \EL^\bot $.
The authors also show how to transform an $\EL^\bot_{gfp}$ base into 
an \EL base. However, it is far from being trivial 
to avoid 
the step of creating an $\EL^\bot_{gfp}$ base in their approach. 

A simplification of the mentioned work has been proposed by 
Borchmann, Distel, and Kriegel (2016) where they show how to mine
$ \EL^\bot $ finite bases with a predefined and fixed  role depth for concept expressions~\cite{Borchmann2016}.
As a consequence, the  base is sound and complete
only w.r.t. CIs containing concepts with bounded role depth.
Their approach avoids the step of creating an $\EL^\bot_{gfp}$ base
but also avoids 
the main challenge
in creating a finite base for \ourEL, 
which is the fact that the role depth
of concepts can be arbitrarily large. 
 
Our work brings together the best of the approaches
by 
Distel (2011)
and Borchmann, Distel, and Kriegel (2016): 
we \emph{directly} compute a finite \ourEL base that captures the \emph{whole} language
(not only up to a certain role depth).
In particular, we present a new approach for
computing the role depth of concepts which \emph{adapts} 
depending on the objects considered during the computation of CIs.
\mypar{Related work} 
Several authors have worked on combining FCA and DLs or on applying methods 
from one field to the other~\cite{anasurvey}.
Baader 
  uses FCA to compute the
subsumption hierarchy of the conjunction of predefined con-
cepts~\cite{Baader95computinga}. 
uses FCA to compute the subsumption hierarchy of the conjunction of predefined concepts.
Rudolph  uses the DL $ \mathcal{FLE} $ for the definition of FCA attributes and FCA techniques for generating a knowledge base~\cite{Rudo04,Rudo06}.
Baader et al. uses FCA for completing missing knowledge in
a DL knowledge base~\cite{BGSS07} .
Baader et al.  proposes a method for building DL ontologies through the interaction of domain experts~\cite{DBLP:conf/iccs/BaaderM00}. 
Sertkaya  presents a survey on applications of FCA methods
in DLs~\cite{Sert10}.
Borchmann and Distel  provide a practical application of the
theory developed by Distel on knowledge graphs~\cite{BoDi11}. 
Borchmann  shows how a base of confident $\ourEL_{gfp}$
concept inclusions can be extracted from a DL interpretation~\cite{Borchmann14}.
Monnin et al.  
compare, using FCA techniques, data present in  DBpedia 
with the constraints of a given ontology to check if the data is compliant 
with it~\cite{DBLP:conf/ismis/MonninLNC17}.
Kriegel~\cite{Kriegel2019} among other contributions
employs FCA notions to build ontologies in logics more expressive than
\ourEL, building upon the framework already established for
\ourEL~\cite{Borchmann2016} and $\EL^\bot_{gfp}$~\cite{Distel2011}.  
He also investigates the same problem for probabilistic DLs~\cite{Kriegel-JELIA19}.

\smallskip


In the next section, we present basic definitions and notation. 
In Section~\ref{sec:mining}, we present the problem of mining 
\ourEL CIs and
establish
lower bounds
for this problem. 
In Section~\ref{sec:adaptable}, we present our main result
for mining \ourEL
 bases with adaptable role depth. 
 Our result uses a notion that relates each vertex in a graph to a set 
 of vertices, 
 called \emph{maximum vertices from} (MVF). In Section~\ref{sec:mvf},
 we show that the MVF of a vertex in a graph 
 can be computed in linear
 time in the size of the graph.
Missing proofs can be found in the long version~\cite{Guimaraes2020}.

\section{Preliminaries}
\label{sec:preliminaries}

%
 \begin{figure}[t]
 	\centering
 	\begin{tabular}{ c | c c c c  }
 		& $\sf{City}$ & $\sf{Region}$ & $ { \sf \exists partof.\top } $ & $ { \sf Settlement } $  \\
 		\hline
 		$ a $ & $ \times $ &  & $ \times $& $ \times $  
 		\\ 
 		$ b $ & $ \times $ &  & $ \times $&$ \times $ 
 \\  
 		$ c $ &  & $ \times $  & $ \times $&
 	\end{tabular}
 	\caption{
 		(a) A dataset 
 		with 
 		 4 attributes
 		 and 3 objects.
 		(b) The 
 		implications $ \sf{City} \rightarrow \sf{\exists partof.\top} $
 		and $ \sf{City} \rightarrow \sf{Settlement}$ hold in the dataset
 		but not $ \sf{City} \rightarrow \sf{Region}$. 
 	}
 	\label{f:formal_context}
 \end{figure}
We introduce the syntax and semantics of \ourEL and 
basic definitions related to description graphs used in the paper. 

\subsection{The Description Logic \ourEL}
\ourEL~\cite{BaBL05} is a lightweight DL,
which only allows for expressing conjunctions and  existential restrictions. 
Despite this rather low expressive power, slight extensions of it have turned 
out to be highly successful in practical applications, especially in the 
medical domain~\cite{SpCC97}. 

We use two \emph{finite} and \emph{disjoint} 
sets, $\NC$ and $\NR$, of \emph{concept} and \emph{role} 
names to define the syntax and semantics of \ourEL.
\ourEL \emph{concept expressions} 
are built according to the 
grammar rule
$C,D::= A  \mid \top\mid\bot\mid C\sqcap D\mid \exists r.C $
with $A\in\NC$ and $r\in\NR$.
We write $ \exists r^{n+1}.C $
as a shorthand for $ \exists r.(\exists r^n.C) $
where $\exists r^1.C:=\exists r.C$.
An \ourEL  \emph{TBox} is a finite set 
of  \emph{concept inclusions} (CIs) $C\sqsubseteq D$, where $C,D$ are \ourEL concept expressions.
We may omit `\ourEL' when we speak of concept expressions, 
CIs, and TBoxes, if this is clear from the context.
We may write $C \equiv D$ (an equivalence) as a short hand for 
when we have both $C\sqsubseteq D$ and $D\sqsubseteq C$.
%
The \emph{signature} of a  concept expression, a CI,
or  a TBox is the set of concept and role names occurring 
in it.

The semantics of  \ourEL is based on \emph{interpretations}. 
An interpretation \Imc is a pair $(\Delta^\Imc,\cdot^\Imc)$ 
where $\Delta^\Imc$ is a non-empty set, called the \emph{domain of \Imc}, and $\cdot^\Imc$ is a function 
mapping each $A\in\NC$ to a subset $A^\Imc$ of $\Delta^\Imc$
and each $r\in\NR$ to a subset $r^\Imc$ of $\Delta^\Imc\times \Delta^\Imc$. 
The function $\cdot^\Imc$ extends to arbitrary \ourEL
concept expressions as usual: 
\begin{gather*}
(C\sqcap D)^\Imc := {}  C^\Imc\cap D^\Imc \quad  (\top)^\Imc:=\Delta^\Imc \quad  (\bot)^\Imc:=\emptyset \\
(\exists r.C)^\Imc := {}  \{x\in\Delta^\Imc\mid  (x,y)\in r^\Imc  \text{ and } y\in C^\Imc  \}
\end{gather*}
An interpretation \Imc \emph{satisfies} a CI $C\sqsubseteq D$, in symbols $\Imc\models C\sqsubseteq D$, iff 
$C^\Imc \subseteq D^\Imc$. 
It satisfies a TBox $\Tmc$ if it satisfies
all CIs in \Tmc.
%
%
%
A TBox \Tmc \emph{entails} a CI $C\sqsubseteq D$, written
$\Tmc\models C\sqsubseteq D$, iff all interpretations 
satisfying \Tmc also satisfy $C\sqsubseteq D$.
We write $\Sigma_\Imc$ for the set
of concept or role names $X$ 
such that  $X^\Imc\neq \emptyset$.
A \emph{finite interpretation} is an interpretation
with a finite domain. 
 
\subsection{Description Graphs, Products, and Unravellings}
We also use the notion of description graphs~\cite{Baader2003}.
The \emph{description graph} $ \dgraph(\I) =(V_\I,E_\I,L_\I) $ of an interpretation \I 
 is defined as (e.g. Figure~\ref{f:descriptiongraph}):
 \begin{enumerate}
\item $ V_\I = \Delta^\I $;
\item $ E_\I = \{ (x,r,y) \mid r\in \NR \text{ and } (x,y)\in r^\I \} $;
\item $ L_\I(x) = \{ A\in \NC\mid x \in A^\I \} $.
 \end{enumerate}

The \emph{description tree} of an \ourEL concept expression $C$
over the signature $\Sigma$ is the finite directed 
tree \(\dgraph(C) = (V_C, E_C, L_C)\) where \(V_C\) is the set of nodes,  
\(E_C \subseteq V_C \times \NR \times V_C\) is the set of edges, and 
\(L_C : V \to 2^{\NC}\) 
is the labelling function.
$ \dgraph(C) $ is  defined inductively: 
\begin{enumerate}
\item for $ C = \top $, $ V_C = \{ \rho_C \} $ and $ L_C(\rho_C) = \emptyset $
where $\rho_C$ is the root node of the tree;
\item for $ C = A\in\NC $, $ V_C = \{ \rho_C \} $ and $ L_C(\rho_C) = A $;
\item for $ C = D_1\sqcap D_2 $, $\dgraph(C)$ is obtained by merging the roots 
$ \rho_{D_1} $, $ \rho_{D_2} $ in one $ \rho_C $ with 
$ L_C(\rho_C) = L_{D_1}(\rho_{D_1}) \cup  L_{D_2}(\rho_{D_2})$;
\item for $ C = \exists r.D $, $\dgraph(C)$ is built from $\dgraph(D)$ by adding 
a new node (root) $ \rho_C $ to $ V_D $ and an edge $ (\rho_C, r, \rho_D) $ to $ E_D $.
\end{enumerate}
The \emph{concept expression} (unique up to logical equivalence) 
$ \C(\dgraph_v) $ of a tree shaped graph
$ \dgraph_v = (V,E,L) $
rooted in $ v $
is
\[
\bigsqcap_{i = 1}^{k} P_i \sqcap \bigsqcap_{j = 1}^l \exists
r_j.\C(\dgraph_{w_j}),
\]
where 
\(L(v) = \{P_i \mid 1 \leq i \leq k\}\), 
$(v,r_j,w_j)\in E$ (and there are $l$ such edges)
and 
 $ \C(\dgraph_{w_j})$ is inductively defined,
with $ \dgraph_{w_j} $ being the subgraph of $\dgraph$ rooted in $w_j$.

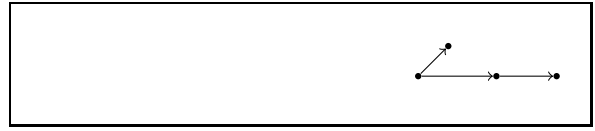
\begin{figure}[btp]
	\centering
	\fbox{
		\begin{tikzpicture}
		\tikzset{d/.style={circle,fill=black,inner sep=0pt,minimum size=3pt}}
		\tikzset{l/.style={fill=white,opacity=0,text opacity=1,align=left}}
		\tikzset{my loop/.style =  {to path={
					\pgfextra{\let\tikztotarget=\tikztostart}
					[looseness=12,min distance=10mm]
					\tikz@to@curve@path},font=\sffamily\small
		}}
		
		\def\x{0}
		\def\y{-.5}
		\node[l,scale=0.9] () at (\x-2.1,\y+1) 
		{$ \sf{City} \sqcap  \exists {\sf government}. {\sf Party} \sqcap
			$ \\ $\exists \sf{partof}.(\sf{Region} \sqcap \exists \sf{capital}.\top)$ };
		
		\def\x{1.6}
		\def\y{0.5}
		\node () [] at (\x,\y){
			\begin{tikzpicture}[scale=0.8, every node/.style={scale=0.8}]
			\def\x{-4.2}
			\def\y{-.9}
			\node[d] (1) at (\x,\y) {};
			\node[d] (2) at (\x+.5,\y+.5) {};
			\node[d] (3) at (\x+1.3,\y) {};
			\node[d] (4) at (\x+2.3,\y) {};
			\node[l] () at (\x-0.7,\y) {$\sf{City}$};
			\node[l] () at (\x+1,\y+0.7) {${\sf Party}$};
			\node[l] () at (\x+1.3,\y+0.25) {$\sf{Region}$};
			\draw[->] (1) to node[l,above,shift={(-0.5,-0.1)}] {$ {\sf govern.} $} (2);
			\draw[->] (1) to node[l,above,shift={(0,-0.55)}] {$ \sf{partof} $} (3);
			\draw[->] (3) to node[l,above,shift={(-0,-0.55)}] {$ \sf{capital} $} (4);
			\end{tikzpicture}
		};
		
%
		%
		
	\end{tikzpicture}
}
\caption{A concept expression and its description graph.}
\label{f:descriptiontree}
\end{figure}

%
A \emph{walk} in a description graph
$\dgraph = (V, E, L)$ between two nodes $u, v \in V$ is a word
$\dgpath{w} = v_0 r_0 v_1 r_1 \ldots r_{n-1} v_n$ where
$v_0 = u$, $v_n = v$, $v_i \in V$, $r_i \in \NR$ and 
$(v_i, r_i, v_{i+1}) \in E$ for all $i \in \{0, \ldots, n-1\}$.
The length of $\dgpath{w}$ in this case is $n$, in symbols, \(|\dgpath{w}| = n\). 
Walks with length  $n=0$ are possible, it means that the walk has just one 
vertex (no edges). Vertices and edges may occur multiple times in a walk.
Let \(\dgraph = (V, E, L)\) be an \(\ourEL\) description graph with $x\in V$ and 
\(d \in \mathbb{N}\). Denote by \(\delta(\dgpath{w})\)  the last vertex in
the walk \(\dgpath{w}\). The \emph{unravelling} of
\(\dgraph\) up to depth \(d\) is the description graph \(\dgraph^x_d
= (V_d, E_d, L_d)\) starting at node $x$ defined as follows:
\begin{enumerate}
	\item 
\(V_d\) is the set of all directed walks in \(\dgraph\) that
start at \(x\) and have length at most \(d\);
	\item 
\(E_d = \{(\dgpath{w}, r, \dgpath{w}rv) \mid v \in V, r \in
\NR,  \dgpath{w}, \dgpath{w}rv \in V_d\}\);
	\item 
\(L_d(\dgpath{w}) = L(\delta(\dgpath{w}))\).  
\end{enumerate}
A \emph{path} is a walk
where vertices do not repeat.
\begin{figure}[hbtp]
\centering
\fbox{
	\begin{tikzpicture}
	\tikzset{d/.style={circle,fill=black,inner sep=0pt,minimum size=3pt}}
	\tikzset{l/.style={fill=white,opacity=0,text opacity=1,align=left}}
	\tikzset{my loop/.style =  {to path={
				\pgfextra{\let\tikztotarget=\tikztostart}
				[looseness=12,min distance=10mm]
				\tikz@to@curve@path},font=\sffamily\small
	}}

	\def\x{0}
	\def\y{0}
	\node[l] () at (\x+.2,\y+1) {$ (i)$};
	\node[] (1) at (\x+1,\y+0.8) {$ a $};
	\node[] (2) at (\x+1,\y) {$ b $};
	\node[l] () at (\x+1,\y+1.05) {$\sf{City}$};
	\node[l] () at (\x+1,\y-0.35) {$\sf{Region}	$};
	\draw[->] (1) to[in=150,out=210,looseness=1.2] 
	node[l,above,shift={(-0.1,-0.2)}] {$ 1 $} (2);
	\draw[->] (2) to[in=330,out=45,looseness=1.2] 
	node[l,above,shift={(+0.1,-0.25)}] {$ 2 $} (1);
	%
	
	\def\x{1.7}
	\def\y{0}
	\node[l] () at (\x+1.3,\y+1) {$ (ii)$};
	\def\y{0.3}
	\node[] (1) at (\x+.5,\y) {$ a $};
	\node[] (2) at (\x+2,\y) {$ a1b $};
	\node[l] () at (\x+.5,\y+0.3) {$\sf{City}$};
	\node[l] () at (\x+2,\y+0.3) {$\sf{Region}$};
	\draw[->] (1) to node[l,above,shift={(-0,-0.4)}] {$ 1 $} (2);
	
	\def\x{4.5}
	\def\y{0}
	\node[l] () at (\x+1,\y+1) {$(iii)$};
	\def\y{0.3}
	\node[] (1) at (\x+.5,\y) {$ a $};
	\node[] (2) at (\x+2,\y) {$ a1b $};
	\node[] (3) at (\x+2.5,\y-0.8) {$ a1b2a $};
	\node[l] () at (\x+.5,\y+0.3) {$\sf{City}$};
	\node[l] () at (\x+2,\y+0.3) {$\sf{Region}$};
	\node[l] () at (\x+2.3,\y-0.5) {$\sf{City}$};
	
	\node[l] () at (\x+.5,\y-.5) {$\sf{Region}$};
	\node[] (4) at (\x+.5,\y-.8) {$ a1b2a1b $};
	\draw[->] (2) to[in=35,out=350,looseness=1.2] 
	node[l,above,shift={(+0.3,-0.3)}] {$ 2 $} (3);
	\draw[->] (1) to node[l,above,shift={(-0.1,-0.4)}] {$ 1 $} (2);
	\draw[->] (3) to node[l,above,shift={(-0.1,-0.4)}] {$ 1 $} (4);

	\end{tikzpicture}
}
\caption{Unravelling of the description graph of the interpretation \I 
	in ~$(i)$.
	For readability, $ \sf{partof} $ has been replaced with symbol $ 1 $
	and $ \sf{capital} $ with symbol $ 2 $.
	$(ii)$ depicts  $\dgraph(\I)^a_1$ and
	$(iii)$ depicts $ \dgraph(\I)^a_3 $.}
\label{f:unravelling}
\end{figure}
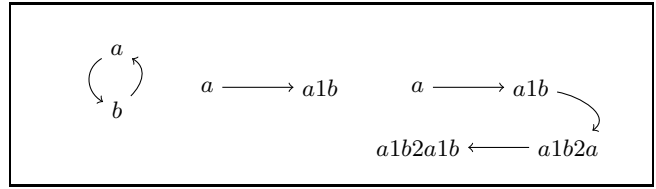

 Let 
 \(\dgraph_1, \dots,
\dgraph_n\)
 be \(n\) description
graphs such that \(\dgraph_i = (V_i, E_i, L_i)\).
Then the \emph{product} of \(\dgraph_1, \dots,
\dgraph_n\) is the description graph \((V, E, L)\)
defined as: 
\begin{enumerate}
	\item 
	\(V = \bigtimes_{i = 1}^{n} V_i\);
	\item 
    \(E = \{((v_1, \dots, v_n), r, (w_1, \dots, w_n)) \mid
            r \in \NR, (v_i, r, w_i) \in E_i, \text{ for all } 1 \leq i \leq
        n\}\);
	\item 
	\(L(v_1, \dots, v_n) = \bigcap_{i=1}^n L_i(v_i)\).
\end{enumerate}

If each \(\dgraph_i\) is a tree with root \(v_i\) then
we denote by \(\prod_{i = 1}^n \dgraph_i\) the tree rooted in
\((v_1, \dots, v_n)\) contained in the product graph of \(\dgraph_1, \dots,
\dgraph_n\).

\section{Mining \ourEL Bases }
\label{sec:mining}

The set of all \ourEL CIs that are satisfied by an interpretation $\I$ is in 
general infinite because whenever $\I \models C \sqsubseteq D$, 
$\I \models \exists r. C \sqsubseteq \exists r. D$ as well.  Therefore one is 
interested in a finite and small set of CIs that
entails the whole set of valid CIs. For mining such a set of CIs from a given
interpretation we employ ideas from FCA and recall literature results.
%

\begin{definition}\label{d:base}
	A TBox \Tmc is a \emph{base} for a finite interpretation \Imc and a DL language $L$,
	if for    every CI $C\sqsubseteq D$, formulated within $L$ and $\Sigma_\I$:
		$\Imc\models C\sqsubseteq D$ iff $\Tmc \models C\sqsubseteq D$.
\end{definition}

%
%

\begin{figure}[btp]
	\centering
	\fbox{
		\begin{tikzpicture}
		\tikzset{d/.style={circle,fill=black,inner sep=0pt,minimum size=3pt}}
		\tikzset{l/.style={fill=white,opacity=0,text opacity=1,align=left}}
		\tikzset{my loop/.style =  {to path={
					\pgfextra{\let\tikztotarget=\tikztostart}
					[looseness=12,min distance=10mm]
					\tikz@to@curve@path},font=\sffamily\small
		}}
		
		%
		\def\x{-2.8}
		\def\y{-.6}
		\node () [] at (\x,\y){
			\begin{tikzpicture}[scale=0.9, every node/.style={scale=0.9}]
			\def\x{-4.2}
			\def\y{-.9}
			\node[] (1) at (\x,\y) {$ x_1 $};
			\node[] (2) at (\x+.5,\y+.8) {$ x_3 $};
			\node[] (3) at (\x+1.3,\y) {$ x_5 $};
			\node[] (4) at (\x+2.6,\y) {$ x_6 $};
			\node[l] () at (\x-0.8,\y) {$\sf{City}$};
			\node[l] () at (\x+1.8,\y+0.9) {${\sf Party, Liberal}$};
			\node[l] () at (\x+1.3,\y+0.3) {$\sf{Region}$};
			\draw[->] (1) to node[l,above,shift={(-0.6,-0.2)}] {$ {\sf govern.} $} (2);
			\draw[->] (1) to node[l,above,shift={(0,-0.6)}] {$ \sf{partof} $} (3);
			\draw[->] (3) to node[l,above,shift={(-0,-0.55)}] {$ \sf{capital} $} (4);
			\end{tikzpicture}
		};
		
		\def\x{1.1}
		\def\y{-.6}
		\node () [] at (\x,\y){
			\begin{tikzpicture}[scale=0.9, every node/.style={scale=0.9}]
			\def\x{-4.2}
			\def\y{-.9}
			\node[] (1) at (\x,\y) {$ x_2 $};
			\node[] (2) at (\x+.55,\y+.8) {$ x_4 $};
			\node[] (3) at (\x+1.3,\y) {$ x_7 $};
			\node[l] () at (\x-0.8,\y) {$\sf{City}$};
			\node[l] () at (\x+1.75,\y+0.8) {${\sf Party }, $ \\ $\sf{Organization}$};
			\node[l] () at (\x+2.1,\y+0.1) {$\sf{Region}$};
			\draw[->] (1) to node[l,above,shift={(-0.6,-0.2)}] {$ {\sf govern.} $} (2);
			\draw[->] (1) to node[l,above,shift={(0.1,-0.6)}] {$ \sf{partof} $} (3);
			\draw[->] (3) to[in=270,out=330,looseness=1.3] node[l,above,shift={(-0,-0.45)}] {$ \sf{capital} $} (1);
			\end{tikzpicture}
		};

	\end{tikzpicture}
}
\caption{Description graph of the interpretation $ \I= \{ \{ x_1,\cdots,x_7 \} ,\cdot^\I \} $ where $ \{ x_1,x_2 \} = \sf{City}^\I $, $ \{ x_3,x_4 \} = \sf{Party}^\I $, 
	$ \{ (x_1,x_5), (x_2,x_7) \} = \sf{partof}^\I $, etc.}
\label{f:descriptiongraph}
\end{figure}
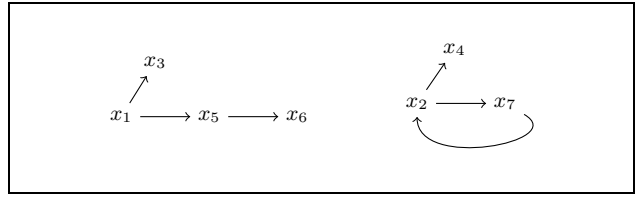

We say that a DL has the \emph{finite base property} (FBP) if, for all finite 
interpretations \I, there is a finite base
with CIs formulated within the DL language and $\Sigma_\I$. 
Not all DLs have the finite base property. 
Consider for instance the fragments $\ourEL_{rhs}$ (and $\ourEL_{lhs}$) 
of $\ourEL$ that allows only 
concept names on the left-hand (right-hand) side but \emph{complex} \ourEL concept expressions on the 
right-hand (left-hand) side of CIs.

\begin{restatable}{proposition}{nofinitebase}

	\label{p:nofinitebase}
	$ \ourEL_{rhs} $ and $ \ourEL_{lhs} $
	do not have the FBP. 
\end{restatable}
\begin{proof}(Sketch)
	No  finite base $ \ourEL_{rhs} $ exists for the interpretation in 
	Figure~\ref{f:exfinitebases}~$ (i) $. 	For every $ n\geq 1 $, 
	the $ \ourEL_{rhs} $ base should entail the CI $ A\sqsubseteq \exists r^n.\top $.
	Similarly, no finite $ \ourEL_{lhs} $ base exists for 
	the interpretation in Figure~\ref{f:exfinitebases}~$ (ii) $.
	For every $ n\geq 1 $, the $ \ourEL_{lhs} $ base should entail the 
	CI $ \exists s.\exists r^n. B \sqsubseteq A $.
\end{proof}


\begin{figure}[hbtp]
\centering
\fbox{
	\begin{tikzpicture}
	\tikzset{d/.style={circle,fill=black,inner sep=0pt,minimum size=3pt}}
	\tikzset{l/.style={fill=white,opacity=0,text opacity=1}}
	\tikzset{my loop/.style =  {to path={
				\pgfextra{\let\tikztotarget=\tikztostart}
				[looseness=12,min distance=10mm]
				\tikz@to@curve@path},font=\sffamily\small
	}}
	
	\def\x{0}
	\def\y{0}
	\node[l] () at (\x-0.2,\y+1) {$(i)$};
	\node[d] (1) at (\x,\y) {};
	\node[d] (2) at (\x+.5,\y+.5) {};
	\node[l] () at (\x-0.22,\y) {$A$};
	\draw[->] (2) to[in=150,out=30,looseness=13.8] node[l,above] {$ r $} (2);
	\draw[->] (1) to node[l,above,shift={(-0.1,-0.1)}] {$ r $} (2);
	
	\def\x{3}
	\def\y{0}
	\node[l] () at (\x-0.2,\y+1) {$(ii)$};
	\node[d] (1a) at (\x,\y) {};
	\node[d] (2a) at (\x+.5,\y+.5) {};
	\node[l] () at (\x-0.22,\y) {$A$};
	\node[l] () at (\x+.75,\y+.4) {$B$};
	\draw[->] (2a) to[in=150,out=30,looseness=13.8] node[l,above] {$ r $} (2a);
	\draw[->] (1a) to node[l,above,shift={(-0.1,-0.1)}] {$ s $} (2a);
	
	\def\x{4.5}
	\def\y{0}
	\node[d] (1b) at (\x,\y) {};
	\node[d] (2b) at (\x+.5,\y+.5) {};
	\draw[->] (2b) to[in=150,out=30,looseness=13.8] node[l,above] {$ r $} (2b);
	\draw[->] (1b) to node[l,above,shift={(-0.1,-0.1)}] {$ s $} (2b);
	\draw[->] (1b) to[out=190,in=290,looseness=1] node[l,below,shift={(-0.1,+0.1)}] {$ r $} (2a);
	
	\end{tikzpicture}
}
\caption{Lack of the FBP   for $\ourEL_{rhs}$ $ (i) $ and
 $\ourEL_{lhs}$ $ (ii) $.}
\label{f:exfinitebases}
\end{figure}

The main difficulty in  creating an \ourEL base is knowing how 
to define the role depth of concept expressions in the base.
In a finite interpretation, an arbitrarily large role depth means
the presence of a cyclic structure in the interpretation. 
However, \ourEL concept expressions cannot 
express cycles.  
The difficulty can be overcomed by extending \ourEL with greatest fix-point semantics.
It is known that the resulting DL, called $\EL^\bot_{gfp}$, has the FBP~\cite{BaDi08,Distel2011}.
The authors then show how to transform an $\EL^\bot_{gfp}$ base into an \ourEL base,
thus,  establishing that \ourEL also enjoys the FBP.

In the following, we show that, although finite, the role depth of 
a base for \ourEL and a (finite) interpretation \I
can be exponential in the size of \I.

\begin{example}\label{e:exponentiald}
		Consider   \I represented in the shaded area in Figure~\ref{f:exponentiald}.
	For $ p_1 = 2, p_2 = 3, p_3 = 5 $ and for all $ k \in  \mathbb{N}^+$,
	we have that $ x_i \in (\exists r^{k\cdot p_i -1}.A)^\I $,
	where $ 1\leq i \leq 3 $.
	We know 
	that 
	$ 30= min(\bigcap_{i=1}^3 \{ k\cdot p_i \mid k\in \mathbb{N}^+ \} ) = \prod_{i=1}^{n} p_i
	$ (which is the least common multiple).
	We also know that for any $ n,p\in\mathbb{N}^+ $,
	$ n+1 $ is a multiple of  $ p $ iff
	$ n $ is a multiple of $p$ minus $1$.
	%
	%
	Therefore, the number 
	$$ d = min(\bigcap_{i=1}^3 \{ k\cdot p_i -1 \mid k\in\mathbb{N}^+ \} ),$$ 
	such that 
	$ \{x_1,x_2, x_3 \} =  B^\I = (\exists r^d. A)^\I$, 
	is 
	$ \prod_{i=1}^{3} p_i -1  = 29$.
	A base for \I should have the CI with role depth at 
	least $ d $ because it has to
	entail the CI  $ B \sqsubseteq \exists r^{d}.A$.
\end{example}

\begin{restatable}{theorem}{exponentialroledepth}
	\label{th:exponentialroledepth}
	There is a finite interpretation $ \I = (\Delta^\I,\cdot^\I) $ 
	such that
	 any 
	\ourEL base for \I  has 
	a concept expression with role depth exponential in the size of \I.
\end{restatable}
\begin{proof}(Sketch)
	We can generalise  Example~\ref{e:exponentiald} to the case where we have an interpretation \Jmc
	that for an arbitrary $ n > 1 $, and for every $ i\in \{1,\cdots,n\} $ and $ k\in \mathbb{N}^+ $,
	there is an $ x\in \Delta^\Jmc $ that satisfies $ x\in (\exists r^{k\cdot p_i -1}.A)^\Jmc $
	where $ p_i $ is the $ i $-th prime number.
	In this case, the minimal role depth of concepts in any base for \Jmc must be $ d \geq  \prod_{i=1}^{n} p_i -1 \geq 2^n$.
\end{proof}

\begin{figure}[hbtp]
	\centering
	\fbox{
		\begin{tikzpicture}
		\tikzset{d/.style={circle,fill=black,inner sep=0pt,minimum size=3pt}}
		\tikzset{l/.style={fill=white,opacity=0,text opacity=1,align=left}}
		\tikzset{my loop/.style =  {to path={
					\pgfextra{\let\tikztotarget=\tikztostart}
					[looseness=12,min distance=10mm]
					\tikz@to@curve@path},font=\sffamily\small
		}}

		\def\x{0}
		\def\y{0}
		\node[] (1) at (\x+1,\y+0.8) {$ x_1 $};
		\node[d] (2) at (\x+1,\y) {};
		\node[l] () at (\x+0.9,\y-0.3) {$A$};
		\node[l] () at (\x+.6,\y+1.0) {$B$};
		\draw[<-] (1) to[in=165,out=210,looseness=1.2] 
		node[l,left] {$ r $} (2);
		\draw[<-] (2) to[in=320,out=25,looseness=1.2] 
		node[l,right] {$ r $} (1);
		%
		
		\def\x{1.7}
		\def\y{0}
		\node[] (1) at (\x+1,\y+0.8) {$ x_2 $};
		\node[d] (2) at (\x+1.5,\y) {};
		\node[d] (3) at (\x+0.5,\y) {};
		\node[l] () at (\x+.25,\y-0.2) {$A$};
		\node[l] () at (\x+.7,\y+1.1) {$B$};
		\draw[->] (1) to[in=40,out=0,looseness=1.2] 
		node[l,above,shift={(-0.1,-0.3)}] {$ r $} (2);
		\draw[->] (2) to[in=290,out=240,looseness=1.2] 
		node[l,above,shift={(+0.1,-0.35)}] {$ r $} (3);
		\draw[->] (3) to[in=180,out=140,looseness=1.2] 
		node[l,above,shift={(+0.1,-0.3)}] {$ r $} (1);
		%
		
		\def\x{4}
		\def\y{0}
		\node[] (1) at (\x+1,\y+0.8) {$ x_3 $};
		\node[d] (2) at (\x+2,\y+0.3) {};
		\node[d] (5) at (\x,\y+0.3) {};
		\node[d] (4) at (\x+.5,\y-0.3) {};
		\node[d] (3) at (\x+1.5,\y-0.3) {};
		\node[l] () at (\x-0.25,\y+0.2) {$A$};
		\node[l] () at (\x+.7,\y+1.1) {$B$};
		\draw[->] (1) to[in=80,out=350] 
		node[l,above,shift={(-0.1,-0.3)}] {$ r $} (2);
		\draw[->] (2) to[in=20,out=290] 
		node[l,above,shift={(-0.2,-0.1)}] {$ r $} (3);
		\draw[->] (3) to[in=315,out=225,looseness=1.2] 
		node[l,above,shift={(+0.1,-0.35)}] {$ r $} (4);
		\draw[->] (4) to[in=250,out=160,looseness=1.2] 
		node[l,above,shift={(+0.2,-0.1)}] {$ r $} (5);
		\draw[->] (5) to[in=190,out=110] 
		node[l,above,shift={(+0.2,-0.3)}] {$ r $} (1);
		
		\def\x{.2}
		\def\y{-1.1}
		\draw[dotted,fill=black,fill opacity=0.1] (\x,\y) rectangle ++(6,2.5);
		
		\def\x{5.8}
		\def\y{0}
		\node[] (1) at (\x+1,\y+0.8) {$ x_4$};
		\node[l] () at (\x+.7,\y+1.1) {$B$};
		\path (1) edge [loop below] node {$ r $} (1);
		
		\def\x{6.8}
		\def\y{0}
		\node[] (1) at (\x+1,\y+0.8) {$ x_5$};
		\node[l] () at (\x+.8,\y+1.1) {$A,B$};
		\path (1) edge [loop below] node {$ r $} (1);
		\end{tikzpicture}
	}
	\caption{Description graph of an interpretation \I.
%
		Let $ X = \{x_1,x_2,x_3\} $. 
		For all $ d  <  29$ we have 
		\(x_4 \in \C\left(\prod_{x \in X}\dgraph(\I)^x_d\right)^\I 
		= (B \sqcap \exists r^d.\top)^\I\).
		However, 
		for all $ k \geq 29 $,
		\(x_4 \not\in \C\left(\prod_{x \in X}\dgraph(\I)^x_k\right)^\I \) since 
		$x_4\not\in (\exists r^{29}.A)^\I$. 
%
	}
	\label{f:exponentiald}
\end{figure}

In addition to the role depth of the concept expressions in the base, the size
of the base itself can also be exponential in the size of the data given as input, which is 
a well-known 
result 
in classical FCA~\cite{DBLP:journals/jucs/Kuznetsov04}.
The DL setting is more challenging than classical FCA, and so, this  lower bound also holds 
in the problem we consider. 
In Section~\ref{sec:adaptable}, we 
present our definition of an \ourEL base for a finite interpretation \I 
and
highlight cases in which the role depth is polynomial in the size of \I. 


\section{Adaptable Role Depth}\label{sec:adaptable}

We present in this section our main result which is our strategy to construct \ourEL bases with adaptable role depth. 
To define an \ourEL base,
we use the notion of a model-based most specific concept (MMSC)
up to a certain role depth.
The MMSC plays a key r\^ole 
in the 
computation of a base from a given finite interpretation. 

\begin{definition}
\label{def:model-based}
		An \ourEL concept expression $C$ is a model-based most specific
	concept of $X\subseteq \Delta^\I$ with role depth $d\geq 0$
	iff
	(1) $X \subseteq C^\I$, 
	(2) $C$ has role depth at most $d$, and
	(3) for all \ourEL concept expressions $D$ with role depth at most
	\(d\), if  $X \subseteq D^\I$
	then $\emptyset \models C \sqsubseteq D$.
\end{definition}

For a given $X \subseteq C^\I$ and a role depth $d$ there may be multiple MMSCs
(always at least one~\cite{Borchmann2016}) but they are logically equivalent. 
So
we  write `\emph{the}' MMSC of $X$ with role depth $d$ (in symbols 
\(\mmsc{X}{\I}[d]\)), meaning a representative of such class of concepts. 
As a consequence of
Definition~\ref{def:model-based}, if \(X = \emptyset\) then \(\mmsc{X}{\I}[d]
\equiv \bot\) for any interpretation \(\I\) and \(d \in \mathbb{N}\).

\begin{example}\label{e:mmsc}
	Consider the interpretation $ \I $ in Figure~\ref{f:descriptiongraph} and 
	let $ X = \{ x_1,x_2\}$.  We have 
	that 
	$ \mmsc{X}{\I}[1] \equiv \sf{City \sqcap \exists government. Party \sqcap \exists partof.Region} $.
	With an increasing $k$, the concept expression $\mmsc{X}{\I}[k]$ 
	can become more and more specific.
	Indeed,  
	$ \mmsc{X}{\I}[2] 
	\equiv 
	\mmsc{X}{\I}[1]
	\sqcap \exists \sf{partof.(Region \sqcap \exists capital.\top)}$
	which is more specific than $\mmsc{X}{\I}[1]$.
	However, for any $ k\geq 2 $, we have that
	$\mmsc{X}{\I}[2] \equiv \mmsc{X}{\I}[k]$.
\end{example}

A straightforward (and inefficient) way of computing  
$ \mmsc{\{X\}}{\I}[d] $, for a fixed $ d $, 
 would be conjoining every \ourEL concept expression 
$ C $ (over \(\NC \cup \NR\)) such that $ X \subseteq C^\I $ and the depth of $C$ is bounded by $d$. A more elegant 
method for computing MMSCs is based on the product of description graphs and
unravelling cyclic concept expressions up to a sufficient role depth.

The MMSC can be written as the concept expression obtained from the
product of description graphs of an interpretation~\cite{Borchmann2016}.
Formally, if \(\I = (\Delta^\I, \cdot^\I)\) is a finite
interpretation, \(X = \{x_1, \dots, x_n\} \subseteq \Delta^\I\) and a $ d\geq
0 $, then $ \mmsc{\{X\}}{\I}[d] \equiv \C(\prod_{i=1}^n
\dgraph(\I)^{x_i}_d)$.

%



The interesting challenge is how to
identify the smallest $d$ that 
satisfies the property: if $x \in \mmsc{X}{\I}[d]^\I$, then 
$x \in \mmsc{X}{\I}[k]^\I$ for every $k > d$.
%
In the following, we develop a method for computing MMSCs with a role depth
that is suitable
for building an $\ourEL$ base of the given 
interpretation. This method is based on the already mentioned 
MVF notion, defined as follows.
  

%


\begin{definition}
    \label{def:mvf}
Given a description graph \(\dgraph= (V,E)\) with \(u \in V\), 
we define the \emph{maximum vertices from (or MVF) $u$ in \(\dgraph\)}, 
denoted \(\mvf(\dgraph, u)\), as: 
%
%
    \[ \max \{ \lend(\dgpath{w}) \mid \dgpath{w}\text{
is a walk in } \dgraph \text{ starting at } u\}\]
where  \(\lend(\dgpath{w})\) is the number of
    distinct vertices occurring in \(\dgpath{w}\).
    Additionally, we define the function \(\mmvf\) as follows:
    \[\mmvf(\dgraph) := \max_{u \in V} \mvf(\dgraph, u).\]


\end{definition}


In other words,  MVF measures the maximum number of distinct vertices that
a walk with a fixed starting point can visit in the graph. 

\begin{example}
\label{ex:mvf}
Consider the interpretation \(\I\) in Figure~\ref{f:descriptiongraph}. Any walk
in the description graph of \(\I\) starting at \(x_1\) will visit at most
three distinct vertices (including \(x_1\)).
Although there are four vertices reachable from \(x_1\),
we have that \(\mvf(\dgraph(\I), x_1)=3\).
For the vertex \(x_2\), there are walks of any finite length, but 
we visit at most three distinct vertices, namely, \(x_2, x_4,x_7\),
and \(\mvf(\dgraph(\I), x_2)=3\). 
\end{example}

For computing the MMSC up to a sufficient role depth based on MVF we use the 
following notion of simulation.

\begin{definition}
\label{def:sim1}
    Let \(\dgraph_1 = (V_1, E_1, L_1)\), \(\dgraph_2 = (V_2, E_2,
    L_2)\) be  \(\ourEL\) description graphs and \((v_1, v_2) \in V_1 \times V_2\).
    A relation \(Z \subseteq V_1 \times V_2\) is a \emph{simulation} from \((\dgraph_1,v_1)\)
    to \((\dgraph_2,v_2)\), if 
    (1) \((v_1, v_2) \in Z\),
    (2) \((w_1, w_2) \in Z\) implies \(L_1(w_1) \subseteq L_2(w_2)\), and
    (3) \((w_1, w_2) \in Z\) and \((w_1, r, w_1') \in E_1\) imply
              there is \(w_2' \in V_2\) such that \((w_2, r, w_2')
            \in E_2\) and \((w_1', w_2') \in Z\).
\end{definition}

Simulations can be used to decide whether an individual from an
interpretation domain belongs to the extension of a given concept expression.

\begin{lemma}[\cite{Borchmann2016}] 
\label{lem:memberSimu1}
    Let \(\I\) be an interpretation, 
    let \(C\) be an
    \(\ourEL\) concept expression, 
    and let
    \(\dgraph(C) = (V_C, E_C, L_C)\) be the \(\ourEL\) description
    graph of \(C\) with root $\rho_C$. For every \(x \in \Delta^\I\),
    there is a simulation from \((\dgraph(C), \rho_C)\) to \((\dgraph(\I), x)\)
iff          \(x \in C^\I\).
\end{lemma}

Lemma~\ref{lem:memberSimu1} together with other previous results is used below 
to prove Lemma~\ref{lem:unravelExt}, which is crucial for defining the adaptable
role depth. 
It shows the upper bound on the required role depth of the MMSC. 

\begin{restatable}{lemma}{unravelExt}
\label{lem:unravelExt}
    Let \(\I = (\Delta^\I, \cdot^\I)\) be a finite interpretation 
    and take an
    arbitrary 
    \(X = \{x_1, \dots, x_n\} \subseteq \Delta^\I\), \(x'\in\Delta^\I\), and \(k \in \mathbb{N}\). Let 
	\[
        d = \mvf\left(\prod_{i = 1}^n\dgraph(\I), (x_1, \dots, x_n)\right) \cdot \mvf(\dgraph(\I), x').
    \]
    If \(x' \in \C\left(\prod_{i = 1}^n\dgraph(\I)^{x_i}_d\right)^\I\)
    then \(x' \in \C\left(\prod_{i = 1}^n\dgraph(\I)^{x_i}_k\right)^\I\).
\end{restatable}

\begin{proof}(Sketch)
    We show in the long version~\cite{Guimaraes2020} the following claim.
    
    \begin{restatable}{claim}{dsim}
        \label{cla:dsim}
    For all description graphs 
    \(\dgraph = (V, E, L)\) and \(\dgraph' = (V', E', L')\), all vertices 
    \(v \in V\) and \(v' \in V'\), and 
    \[d = \mvf(\dgraph, v) \cdot \mvf(\dgraph', v')\]
    if there is a simulation \(Z_d : (\dgraph_d^{v}, v) \mapsto (\dgraph'
    , v')\),
    then there is a simulation \(Z_k : (\dgraph_k^{v}, v) \mapsto (\dgraph' v')\) for all \(k \in \mathbb{N}\). 
    \end{restatable}

    If \(k \leq d\), one can  restrict \(Z_d\) to the vertices of
    \(\dgraph_k^{v}\), which would be a subgraph of \(\dgraph_d^{v}\).
    Otherwise, the intuition behind this claim is that the pairs in \(Z_d\) define a walk
    in \(\dgraph'\) for each walk in \(\dgraph\) that has length at most \(d - 1\).
    And if a walk in \(\dgraph\) has
    length at least \(d - 1\), then there is a vertex \(w\) that this walk
    visits twice while the image of this walk in \(\dgraph'\) also repeats
    a vertex at the same time.
    This paired repetition can be used to find a matching vertex in \(V'\) for each
    vertex of \(\dgraph_k^{v}\) by recursively shortening the walk that
    this vertex corresponds to if it has length \(d\) or larger.

    Lemma~\ref{lem:memberSimu1} and \(x' \in \C\left(\prod_{i = 1}^n\dgraph(\I)^{x_i}_d\right)^\I\) 
    imply
    that there is a simulation $Z_d$ from \((\prod_{i
    = 1}^n\dgraph(\I)^{x_i}_d, (x_1, \dots, x_n))\) to
    \((\dgraph(\I), x')\). Then, by Claim~\ref{cla:dsim} there is
    a simulation \(Z_k : (\prod_{i = 1}^n\dgraph(\I)^{x_i}_k, (x_1, \dots, x_n)) \mapsto (\dgraph(\I),x')\)
    (we just need to take \(\dgraph = \prod_{i = 1}^n\dgraph(\I)\),
    \(\dgraph' = \dgraph(\I)\), \(v = (x_1, \dots, x_n)\) and \(v' = x'\)).
    Therefore, Lemma~\ref{lem:memberSimu1} implies that
    \(x' \in \C\left(\prod_{i = 1}^n\dgraph(\I)^{x_i}_k\right)^\I\).
\end{proof}
Lemma~\ref{lem:unravelExt} shows that even for vertices that are parts of
cycles, there is a certain depth of unravellings, which we call a fixpoint,
that is guaranteed to be an upper bound.

Proposition~\ref{prop:mvfprod} gives an intuition about how large the MVF of
a vertex in a product graph can be when compared to the MVF of the corresponding
vertices in the product's factors.

\begin{proposition}
\label{prop:mvfprod}
    Let \(\{\dgraph_i \mid 1 \leq i \leq n\}\) be \(n\) description
    graphs such that \(\dgraph_i = (V_i, E_i, L_i)\). Also let \(v_i \in
    V_i\). Then:
    \[ \mvf\left(\prod_{i = 1}^n \dgraph_i, (v_1, \dots, v_n)\right) \leq \prod_{i = 1}^n \mvf(\dgraph_i, v_i).\]
\end{proposition}

\begin{proof}
    Let \(\dgpath{w}\) be an arbitrary walk in
    \(\prod_{i = 1}^n \dgraph_i, (v_i)_{1 \leq i \leq
    n}\) that starts in \((v_1, \dots, v_n)\) and let \((w_1, \dots, w_n)\) be
    a vertex in this walk. It follows from the definition of product that each \(w_i\) belongs to
    a walk in \(\dgraph_i\) that begins in \(v_i\). Therefore, there are
    only \(\mvf(\dgraph_i, v_i)\) options
    for each \(w_i\). Hence, there are at most \(\prod_{i = 1}^n
    \mvf(\dgraph_i, v_i)\) possible options for \((w_1, \dots, w_n)\). In other
    words, \(\lend(\dgpath{w}) \leq \prod_{i = 1}^n \mvf(\dgraph_i, v_i)\).
    Since \(\dgpath{w}\) is arbitrary, we can conclude that
    \[\mvf\left(\prod_{i = 1}^n \dgraph_i, (v_1, \dots, v_n)\right) \leq
    \prod_{i = 1}^n \mvf(\dgraph_i, v_i).\] 
\end{proof}

Although the MVF of a product can be exponential in \(|\Delta^\I|\), there are many cases in which it
is 
linear in \(|\Delta^\I|\). Example~\ref{ex:boundedwalks} illustrates one such case.

\begin{example}
\label{ex:boundedwalks}
Consider the interpretation of Figure~\ref{f:descriptiongraph}.
The elements \(x_1, x_3, x_4, x_5\) and \(x_6\) never reach cycles, therefore,
each of them can only have walks up to a finite length.
Take \(X = \{x_1, x_2\}\). Since every walk in \(\dgraph(\I)\) starting from
\(x_1\) has length at most \(2\), the longest walk possible in \(\prod_{i \in
\{1, 2\}}\dgraph(\I)\) starting at node \((x_1, x_2)\) is:
\((x_1, x_2), {\sf{partof}} , (x_5, x_7), {\sf{capital}}, (x_6, x_2)\). Thus
\[\mvf\left(\prod_{i \in \{1, 2\}} \dgraph(\I), (x_1, x_2)\right) = 2.\]
Take \(X = \{x_1, x_7\}\). 
Since
\(x_1\) and \(x_7\) do not share labels in their outgoing edges
\[\mvf\left(\prod_{i \in \{1, 7\}} \dgraph(\I), (x_1, x_7)\right) = 1.\]
\end{example}

The observations about the MVF in Example~\ref{ex:boundedwalks} are
generalised in Lemma~\ref{lem:boundedwalks} 
which shows a sufficient condition
for polynomial (linear) role depth.

\begin{restatable}{lemma}{boundedwalks}
\label{lem:boundedwalks}
    Let \(\I = (\Delta^\I, \cdot^\I)\) be a finite interpretation
    and \(X = \{x_1, \dots, x_n\} \subseteq \Delta^\I\). If for some \(1
    \leq i \leq n\) it holds that every 
    walk in \(\dgraph(\I)\) starting at \(x_i\) has length at most \(m\) for
    some \(m \in \mathbb{N}\), then 
    \(\mvf\left(\prod_{i = 1}^n \dgraph(\I), (x_1, \dots, x_n)\right) \leq \mvf\left(\dgraph(\I), x_i\right)\).
\end{restatable}

\begin{proof}(Sketch)
    As it happens in Example~\ref{ex:boundedwalks}, it can be proven that
    whenever there is a vertex \(x_i\) for which every walk starting at it 
    has length at most \(m\), then $m$ also bounds 
		the lengths of the walks starting at \((x_1, \dots x_n)\) in  $\prod_{i
    = 1}^n \dgraph(\I)$.
\end{proof}

Combining the bounds for the fixpoint and MVF given by 
Lemmas~\ref{lem:unravelExt} and~\ref{lem:boundedwalks}, we can define a 
function that returns  an upper approximation of the fixpoint, for any subset 
of the domain of an interpretation, as follows.

\begin{definition}
\label{def:newd}
    Let \(\I = (\Delta^\I, \cdot^\I)\) be a finite interpretation 
    and \(X = \{x_1, \dots, x_n\} \subseteq \Delta^\I\).  Also let 

    \begin{align*}
        &X_{lim} = \{x \in X \mid \exists m \in \mathbb{N} : \text{every walk} \\
        &\qquad \text{starting at } x \text{ in } \dgraph(\I) \text{ has length} \leq m\}.
    \end{align*}

    The function \(\newdsym_\I : \powerset{\Delta^\I} \mapsto
    \mathbb{N}\) is defined as follows:

    \[
        \newd{X}{\I} = 
        \begin{cases} 
            d - 1 & \text{if } X_{lim} \neq \emptyset\\

            d \cdot \mmvf(\dgraph(\I)) & \text{otherwise},
        \end{cases}
    \]
    where \(d = \mvf\left(\prod_{i=1}^n \dgraph(\I), (x_1, \dots,
    x_n)\right)\).
\end{definition}


Next, we prove that function \(\newdsym_\I\) is indeed an upper bound for the
fixpoint of an MMSC. The idea sustaining Lemma~\ref{lem:mmscK} is that if \(x \in X \subseteq \Delta^\I\)
and every walk in \(\dgraph(\I)\) starting at \(x\) has length at most \(m\), then \(m\)
can be used as a fixpoint depth for the MMSC of \(X\) in \(\I\). 
Lemma~\ref{lem:unravelExt} covers the cases where vertices are the starting
point of walks of any length.

\begin{restatable}{lemma}{mmscK}
\label{lem:mmscK}
    Let \(\I = (\Delta^\I, \cdot^\I)\) be a finite interpretation 
    and \(X \subseteq \Delta^\I\). Then, for any \(k \in \mathbb{N}\), it holds
    that:
    \[\mmsc{X}{\I}[\newd{X}{\I}]^\I \subseteq \mmsc{X}{\I}[k]^\I.\]
\end{restatable}

\begin{proof}(Sketch)
    Let \(X = \{x_1, \dots, x_n\} \subseteq \Delta^\I\).  If \(k \leq
    \newd{X}{\I}\), the lemma holds trivially.  For \(k > \newd{X}{\I}\) we
    divide the proof in two cases. First, if there is a \(x_i \in X\) such that
    every walk in \(\dgraph(\I)\) starting at \(x_i\) has length at most \(m\)
    for some \(m \in \mathbb{N}\), then as stated in
    Lemma~\ref{lem:boundedwalks}, every walk in \(\prod_{i=1}^n
    \dgraph(\I)\) starting at \((x_1, \dots, x_n)\) has length at most
    \(\mvf\left(\prod_{i=1}^n \dgraph(\I), (x_1, \dots, x_n)\right)
    - 1\).

    In other words, even when \(k > \newd{X}{\I}\), we have: \[\prod_{i=1}^n
    \dgraph(\I)_k^x = \prod_{i=1}^n \dgraph(\I)_\newd{X}{\I}^x\], and
    therefore, we can apply Lemma~\ref{lem:memberSimu1} to conclude that:
    \(\mmsc{X}{\I}[\newd{X}{\I}]^\I \subseteq \mmsc{X}{\I}[k]^\I\).  
    Otherwise, if \(X_{lim} \neq \emptyset\), the lemma is a direct consequence
    of Definition~\ref{def:newd} and Lemma~\ref{lem:unravelExt}.
\end{proof}

In the remaining of this paper, we write \(\mmsc{X}{\I}\)
as a shorthand for \(\mmsc{X}{\I}[\newd{X}{\I}]\). 
An important consequence of Lemma~\ref{lem:mmscK} and the definition of MMSC is
that, for any \(\ourEL\) concept expression \(C\) and finite interpretation
\(\I\), it holds that \(C^\I = \mmsc{C^\I}{\I}^\I\).

\begin{lemma}
\label{lem:r6p}
    Let \(\I = (\Delta^\I, \cdot^\I)\) be a finite interpretation.
    Then, for all \(\ourEL\) concept expression
    \(C\) it holds that: \(\mmsc{C^\I}{\I}^\I = C^\I\).
\end{lemma}

\begin{proof}
    Direct consequence of Lemma 4.4 (vi) 
    of \cite{Borchmann2016} and
    Lemma~\ref{lem:mmscK}.
\end{proof}

We use this result below to define a
finite set of concept expressions \(M_\I\)  
for building a base of the CIs valid in \I.
\begin{definition}
\label{def:setofconceptexpressions:MI}
    Let \(\I = (\Delta^\I, \cdot^\I)\) be a finite interpretation. 
    The set \(M_\I\) is the union of $\{\bot\} \cup \NC$ and 
    \[
     \{\exists r.\mmsc{X}{\I} \mid r \in \NR
    \text{ and } X \subseteq \Delta^\I, X \neq \emptyset\}\]
    We also define
    \(\Lambda_\I = \left\{ \bigsqcap U \mid U \subseteq M_\I \right\}\).
\end{definition}
Building the base mostly relies on the fact that,
given a finite interpretation \(\I\), for any \(\ourEL\) concept expression
\(C\), there is a concept expression \(D \in \Lambda_\I\) such that \(C^\I
= D^\I\).
\begin{restatable}{theorem}{thmbase}
\label{thm:b2}
    Let \(\I\) be a finite interpretation 
    and let $ \Lambda_\I $ be 
    defined as above. Then, 
%
    \begin{align*}
        \base(\I) =&
        \{C \equiv \mmsc{C^\I}{\I} \mid C \in \Lambda_\I\} \cup {}\\
        & \{C \sqsubseteq D \mid C, D \in \Lambda_\I
        \text{ and } \I \models C \sqsubseteq D \}
    \end{align*}
    is a finite \(\ourEL\) base for \(\I\).
\end{restatable}

\begin{proof}(Sketch) 
    As \(\Lambda_\I\) is finite, so is \(\base(\I)\). The CIs
    are clearly sound and the soundness of the equivalences is due to 
    Lemma~\ref{lem:r6p}.
    For completeness, assume that 
    \(\I \models C \sqsubseteq D\).
%
    Using an adaptation of Lemma 5.8 from \cite{Distel2011} and 
		Lemma~\ref{lem:r6p} above,
     we
     can prove, by induction on the structure of the concept expressions \(C\)
     and \(D\),
    that there are concept expressions \(E, F \in \Lambda_\I\) 
    such that $ \base(\I) \models E \equiv \mmsc{C^\I}{\I}$,
    $ \base(\I) \models F \equiv \mmsc{D^\I}{\I}$,
    $ \base(\I) \models C \equiv \mmsc{C^\I}{\I}$,
    and $ \base(\I) \models D \equiv \mmsc{D^\I}{\I}$.
    By construction, as \(E \sqsubseteq
    F \in \base(\I)\),
    we
    can prove that whenever \(\I \models C \sqsubseteq D\), so does
    \(\base(\I)\).
\end{proof}



Recall the interpretation \I in Figure~\ref{f:exponentiald}. 
In order to compute a base for \I, we should compute an MMSC with role depth 
at least $29$.
An important benefit of our approach is that the role depth of the other MMSCs, which 
are part of the mined CIs in the base may be smaller.
For instance, the role depth of  $\mmsc{\{x_1\}}{\I}$   is  
$10$. 
In the next section, we show that one can compute the MVF of a vertex
in a graph in linear time in the size of the graph. 

\section{Computing the MVF}
\label{sec:mvf}

As discussed in Section~\ref{sec:adaptable}, the MVF is the key to provide
an upper bound for the 
fixpoint for each MMSC.
An easy way to estimate the MVF would consist in computing
the number of vertices reachable from \(v\) in the description graph
\(\dgraph\). Let \(\reach(\dgraph,v)\) be such a function.
By definition it holds that \(\mvf(\dgraph, v)
\leq \reach(\dgraph, v)\).
Although \(\reach(\dgraph,v)\) can be computed in polynomial time,
  the difference between these two metrics can
be quite large. For instance, consider that \(v\) is the root of a description
graph \(\dgraph\) that is a binary tree with \(2^n\) nodes. Then \(\mvf(\dgraph,
v) = n\), while \(\reach(\dgraph, v) = 2^n\).

In this section, we present an algorithm to compute \(\mvf(\dgraph, v)\) that takes
linear time in the size of  \(\dgraph\), but first 
we need to
recall some fundamental concepts from Graph Theory, one of them is the notion
of strongly connected components (Definition~\ref{def:scc}).

%

\begin{definition}
\label{def:scc}
    Let \(\dgraph = (V, E, L)\) be a description graph. The \emph{strongly connected
    components} (SCCs) of \(\dgraph\), in symbols \(\SCC(\dgraph)\), are the partitions
    \(V_1, \dots, V_n\) of \(V\) such that for all \(1 \leq i \leq n\): if \(u, v \in V_i\) then
    there is a path from \(u\) to \(v\) and a path from \(v\) to \(u\) in
    \(\dgraph\).
    Additionally, we define a function \(\scc(\dgraph, v)\), which returns the
    SCC of \(\dgraph\) that contains \(v\).
\end{definition}


A compact way of representing a description graph $ \dgraph $ consists in
regarding each SCC in $ \dgraph $ as a single vertex. 
This compact graph is a directed acyclic graph (DAG), also called 
condensation of $ \dgraph $~\cite{Harary1965}, and it is formalised in Definition~\ref{def:condense}.

\begin{definition}
\label{def:condense}
    Let \(\dgraph = (V, E, L)\) be a description graph. The
    \emph{condensation of \(\dgraph\)} is the directed acyclic graph \(\dgraph^* = (V^*,
    E^*)\) where \(V^* = \{\scc(\dgraph, u) \mid u \in V\}\) and \(E^* = \{(\scc(\dgraph, u), \scc(\dgraph, v)) \mid 
    (u, r, v) \in E \text{ and } \scc(\dgraph, u) \neq \scc(\dgraph, v)\}\).
    Also, if \(\dgpath{w}^*\) is path in
    \(\dgraph^*\), the \emph{weight of \(\dgpath{w}^*\)}, in symbols
    \(\weight(\dgraph^*)\), is the sum of the sizes of the
    SCCs that appear as vertices of \(\dgpath{w}^*\).
\end{definition}

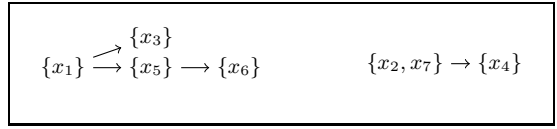
\begin{figure}[btp]
	\centering
	\fbox{
		\begin{tikzpicture}
		\tikzset{d/.style={circle,fill=black,inner sep=0pt,minimum size=3pt}}
		\tikzset{l/.style={fill=white,opacity=0,text opacity=1,align=left}}
		\tikzset{my loop/.style =  {to path={
					\pgfextra{\let\tikztotarget=\tikztostart}
					[looseness=12,min distance=10mm]
					\tikz@to@curve@path},font=\sffamily\small
		}}
		
		%
		\def\x{-2.8}
		\def\y{-.6}
		\node () [] at (\x,\y){
			\begin{tikzpicture}[scale=0.9, every node/.style={scale=0.9}]
			\def\x{-4.2}
			\def\y{-.9}
            \node[] (1) at (\x,\y) {$ \{x_1\} $};
            \node[] (2) at (\x+1.3,\y+.45) {$ \{x_3\} $};
            \node[] (3) at (\x+1.3,\y) {$ \{x_5\} $};
            \node[] (4) at (\x+2.6,\y) {$ \{x_6\} $};
			\draw[->] (1) to node[l,above,shift={(-0.6,-0.2)}] {} (2);
			\draw[->] (1) to node[l,above,shift={(0,-0.6)}] {} (3);
			\draw[->] (3) to node[l,above,shift={(-0,-0.55)}] {} (4);
			\end{tikzpicture}
		};
		
		\def\x{1.1}
		\def\y{-.6}
		\node () [] at (\x,\y){
			\begin{tikzpicture}[scale=0.9, every node/.style={scale=0.9}]
			\def\x{-4.2}
			\def\y{-.9}
            \node[] (1) at (\x,\y) {$ \{x_2, x_7\} $};
            \node[] (2) at (\x+1.4,\y) {$ \{x_4\} $};
            \draw[->] (1) to node[l,above,shift={(-0.6,-0.2)}] {} (2);
			\end{tikzpicture}
		};

	\end{tikzpicture}
}
\caption{\label{f:condense}Condensation of the description graph in
Figure~\ref{f:descriptiongraph}. Every vertex is an SCC of the original graph
and the edges indicate accessibility between the SCCs. Also, the condensation
has no labels.}
\end{figure}

We use these notions to 
link the
MVF (Definition~\ref{def:mvf}) to the paths in the condensation graph
in Lemma~\ref{lem:mvfCondense}.  

\begin{restatable}{lemma}{mvfCondense}
\label{lem:mvfCondense}
    Let \(\dgraph = (V, E, L)\) be a description graph, let \(\dgraph^* = (V^*, E^*)\) be the
    condensation of $ \dgraph $, and \(v \in V\).
    Then:
    \begin{align*}
        \mvf(\dgraph, v) = &\max \left\{\weight(\dgpath{w}^*) \mid \right.\\
        &\left. \qquad \dgpath{w}^* \text{ is a path in } \dgraph^* \text{ starting at } \scc(\dgraph, v)\right\}. 
    \end{align*}
\end{restatable}
\begin{proof}(Sketch)
	First we prove that every path \(\dgpath{w}^* = V_1, \dots, V_m\) in
	\(\dgraph^*\) starting
	at \(\scc(\dgraph, v)\) induces a walk \(\dgpath{w}\) 
	in \(\dgraph\) starting at \(v\) with
	\(\lend(\dgpath{w}) = \weight(\dgpath{w}^*)\).
    Then, we show that   
    if \(\dgpath{w}^*\) has maximal weight, then
	no walk in
	\(\dgraph\) starting at \(v\) can visit more than
	\(\weight(\dgpath{w}^*)\) vertices.
	%
\end{proof}

By Lemma~\ref{lem:mvfCondense},  we only need to compute the
maximum weight of a path in \(\dgraph^*\) that starts at \(\scc(\dgraph^* ,v)\)
to obtain the MVF of
a vertex \(v\) in a description graph \(\dgraph\).
Algorithm~\ref{alg:mvf} relies on this result and proceeds as follows:
first, it computes the SCCs of the description graph and the
condensation graph. Then, the algorithm transverses the condensation graph,
using an adaptation of depth-first search to determine the maximum path weight
for the initial SCC. 

	

	
	
	

\begin{algorithm}
    \DontPrintSemicolon
    \SetAlgoLined
    \LinesNumbered
    \caption{Computing MVF via Lemma~\ref{lem:mvfCondense}}
    \label{alg:mvf}

    \SetKwProg{Fn}{Function}{:}{end}
    \SetKw{null}{null}

    \KwIn{A description graph \(\dgraph = (V, E, L)\) and a vertex \(v \in V\)}
    \KwOut{The MVF of \(v\) in \(\dgraph\), i.e., \(\mvf(G, v)\)}

    \(V^* \gets \SCC(\dgraph)\)\;
    \(E^* \gets \text{condense}(\dgraph, V^*)\)\;
    \(\dgraph^* \gets (V^*, E^*)\)\;
    \For{\(V' \in V^*\)}{
        \({wgt}[V'] \gets \null\)
    }

    \Return{\({\sf maxWeight}(\dgraph^*, \scc(\dgraph, v), {wgt})\)}\;

    \tcp{Auxiliary function}
    \Fn{\({\sf maxWeight}(\dgraph^*, V', {wgt})\)}{%
        \({current}\gets 0\)\;
        \For{\(W' \in \{U' \in V^* \mid (V', U') \in E^*\}\)}{\label{line:successors}
            \eIf{\( {wgt} [W'] = \null\)}{
                \( {current}\gets \max(current, {\sf maxWeight}(\dgraph^*, W', {wgt}))\)
            }{
                \( {current}\gets wgt[W']\)
            }
            \({wgt}[V'] \gets {current} + |V'|\)\;
        }
        \Return{\({wgt}[V']\)}\label{line:recursiveReturn}
    }
\end{algorithm}

Algorithm~\ref{alg:mvf} assumes that the SCCs and condensation are computed
correctly. Besides keeping the computed values, the array \(wgt\) prevents
recursive calls on SCCs that have already been processed.
According to Lemma~\ref{lem:mvfCondense}, to prove that Algorithm~\ref{alg:mvf} is correct we
just need to prove that the function \({\sf maxWeight}\) in fact returns the maximum
weight of a path in the condensation given a starting vertex (which corresponds
to an SCC in the original graph).

%
%
%
%
\begin{restatable}{lemma}{lemMaxWeight}
 	\label{lem:maxWeight}
 	Given   \(\dgraph = (V, E, L)\)   and \(v \in V\) as input, 
 	Algorithm~\ref{alg:mvf} returns 
 	 the  maximum weight of a path in the condensation of \(\dgraph\)
 	starting at \(\scc(\dgraph, v)\).
 \end{restatable}
\begin{proof}(Sketch)
	Let \(\dgraph^* = (V^*, E^*)\) be the condensation of \(\dgraph\). If 
	\(\scc(\dgraph, v)\)
	has no successor in $ \dgraph^* $, then
	the output of $ {\sf maxWeight } $ is correct.
	If \(\scc(\dgraph, v)\) has successors, then the maximum
	weight of a path staring at \(\scc(\dgraph, v)\) in \(\dgraph^*\) is given by \(|\scc(\dgraph, v)|\)
	plus the maximum value computed among its successors. This equation holds
	because \(\dgraph^*\) is a DAG.  
\end{proof}

Lemmas~\ref{lem:mvfCondense} and~\ref{lem:maxWeight} imply that
Algorithm~\ref{alg:mvf} computes the MVF of \(v\) in \(\dgraph\) correctly.
Moreover, the computation of SCCs can be done in time \(O(|V| + |E|)\)
\cite{Tarjan1972}, the condensation in time \(O(|E|)\) \cite{Martello1982} and
the depth-first transversal via \({\sf maxWeight}\) in time \(O(|V| + |E|)\).
Hence, it is possible to compute the MVF of a vertex in a graph in linear time
in the size of the description graph 
even if it consists solely of
cycles. Yet, given an interpretation \(\I = (\Delta^\I, \cdot^\I)\)
the graph given as input to Algorithm~\ref{alg:mvf} might be a product graph with
an exponential number of vertices in \(|\Delta^\I|\). Also, Algorithm~\ref{alg:mvf}
can be modified to compute the MVF for all vertices by starting
the function \({\sf maxWeight}\) from an unvisited SCC until all vertices are
visited in polynomial time in the size of the graph.

 



\section{Conclusion}


In this work, we introduce a way of computing \(\ourEL\) bases from finite
interpretations that adapts the role depth of concepts according to the 
the structure of interpretations. Our definition relies on 
a notion that relates vertices in a graph to sets of vertices, called MVF.
We have also shown that the MVF computation
can be performed in polynomial time in the size of the underlying graph structure. 
%
Our \ourEL base, however, is not minimal. As future work, we plan to build on previous
results combining FCA and DLs to define a base with minimal cardinality. 
%
%
We will also 
investigate the problem of mining CIs in the presence of 
noise in the dataset. We plan to use the support 
and confidence measures from association rule mining to deal with noisy data
 and implement our approach using knowledge graphs as datasets.  
 

\section*{Acknowledgements}

We would like to thank Prof. Sebastian Rudolf for contributing
with the examples in Figure~\ref{f:exfinitebases}. 
Ozaki is supported by the Research Council of Norway, project number 316022. 
Guimar\~{a}es is supported by the ERC project LOPRE (819416), led by Prof. Saket Saurabh.
Parts of this work have been done in the context of CEDAS (Center for Data Science, University of Bergen, Norway).
This work was supported by the Free University of Bozen-Bolzano, Italy,
under the project PACO.


\bibliography{references} 

\ifappendix
\clearpage
\appendix
\section{Proofs for Section~\ref{sec:mining}}



We prove that \(\EL_{rhs}\) and \(\EL_{lhs}\) do not have the finite base
property (Propositions~\ref{th:infinitebaserhs} and~~\ref{th:infinitebaselhs}). 

\begin{proposition}\label{th:infinitebaserhs}
	$ \EL^\bot_{rhs} $ does not have the finite base property.
\end{proposition}
\begin{proof}
	Consider the interpretation $ \I = (\{ x_1,x_2 \}, \cdot^\I) $
	where $ \{(x_1,x_2), (x_2,x_2)\} = r^\I $,
	$ \{x_1 \} = A^\I $ and every other concept and role name is mapped by 
$ \cdot^\I $ to $ \emptyset $ (Figure~\ref{f:exfinitebases}~($ i$)).
	In \I, $   A^\I=\{x_1\} $  and 
	for all $ n\in \mathbb{N}^+ $,
	$ x_1\in (\exists r^n.\top ) ^\I$ .
	Assume that \Bmc is a base for \I and  $ \EL_{rhs} $.
	As \Bmc is a (finite) TBox formulated in $ \EL_{rhs} $ with symbols from $\Sigma_\I$,
it can only have CIs of the form 	$A \sqsubseteq C$.
	Since $\Imc\models A \sqsubseteq \exists r^n.\top$, for all $n\in \mathbb{N}^+$, 
it follows that \Bmc is infinite, which is a contradiction.
\end{proof}

\begin{proposition}\label{th:infinitebaselhs}
	$ \EL^\bot_{lhs} $ does not have the finite base property.
\end{proposition}
\begin{proof}
In this proof, assume that CIs are formulated in $ \EL_{lhs} $.
    Consider the  interpretation 

    \[ \I = (\{ x_1,x_2,x_3,x_4 \}, \cdot^\I)\]
with 
    \begin{align*}
        r^\I &= \{(x_2,x_2),(x_4,x_4),(x_3,x_2)\}\\
        s^\I &= \{(x_1,x_2),(x_3,x_4)\} \\
        A^\I &= \{x_1\}\\
        B^\I &= \{x_2\}
    \end{align*}
    and every other concept and role name  is mapped by $ \cdot^\I $ to $ \emptyset $  (see Figure~\ref{f:exfinitebases}~($ ii $)).

	By definition of \I, for all $n\in\mathbb{N}^+$, we have that
	$\I\models \exists s.\exists r^n.B\sqsubseteq A$. 
So if \Bmc is a base for $ \EL_{lhs} $ and \I then 	$\Bmc\models \exists s.\exists r^n.B\sqsubseteq A$ for all $n\in\mathbb{N}^+$.
	Now, observe that there is no $D$ such that
	\begin{enumerate}
\item	$\emptyset\models \exists s.\exists r^n.B\sqsubseteq D$,
	\item	$\emptyset\nvDash D\sqsubseteq \exists s.\exists r^n.B$ (where $\nvDash$ means `does not entail'),	
	\item and $\I\models D\sqsubseteq A$.
	\end{enumerate}
	The reason for the above is because 
	$x_3\in D^\I$ for all $D$ satisfying (1) and (2) but $x_3\not\in A^\I$. 
	Moreover, there is no $k\in\mathbb{N}^+$ such that $\I\models \exists r^k.B\sqsubseteq B$ 
or $\I\models \exists r^k.B\sqsubseteq A$ (because
$x_3\not\in B^\I$ and $x_2,x_3\not\in A^\I$ but $x_2,x_3\in (\exists r^k.B)^\I$).
	So,  \Bmc can only entail $\exists s.\exists r^n.B\sqsubseteq A$
if there is a CI in \Bmc with a concept equivalent to $\exists s.\exists r^n.B$.
This concept needs to have role depth $n$. Since
$\Bmc\models \exists s.\exists r^n.B\sqsubseteq A$ for all $n\in\mathbb{N}^+$,
there are CIs with role depth $n$ for all $n\in\mathbb{N}^+$. This means that \Bmc cannot be finite.
\end{proof}






Next, we prove the result which shows that the depth of roles in a base has an
exponential lower bound.  

\exponentialroledepth*
\begin{proof}
	For any $ n\geq 1 $, we consider the interpretation \I
	where for every 
	$ i\in \{1,\cdots,n\} $ and $ k\geq 1 $,
	there is  $ x_i\in \Delta^\I $ 
	that satisfies $ x_i\in (\exists r^{k\cdot p_i -1}.A)^\I $, 
	$ x_i \in B^\I $, and $ x_i\not\in (\exists r^{l}.A)^\I $ for $ l\not\in \{k\cdot p_i -1\mid k\geq 1  \} $
	where $ p_i $ is the $ i $-th prime number.
	
	We know 
	that 
	$ \min(\bigcap_{i=1}^n \{ k\cdot p_i \mid k\geq 1 \} ) = \prod_{i=1}^{n} p_i
	$ (which is the least common multiple).
	We also know that for any $ n,p\in\mathbb{N}^+ $,
	$ n+1 $ is a multiple of  $ p $ iff
	$ n $ is a multiple of $ p$ minus $  1 $.
	Therefore,  $ d = min(\bigcap_{i=1}^n \{ k\cdot p_i -1 \mid k\geq 1 \} )$,
	is the minimal number such that 
	$B^\I= (\exists r^d. A)^\I $.
Since $d=\prod_{i=1}^{n} p_i -1 \geq 2^n$,
	the statement holds because 
	a base for \I
	should entail the CI 
	$ B \sqsubseteq \exists r^d. A $.
	For this to happen, it should have a CI with role depth at least $ d $.
\end{proof}

\section{Proofs for Section~\ref{sec:adaptable}}

Now we will prove Claim~\ref{cla:dsim}, which is part of Lemma~\ref{lem:unravelExt} that underlies 
our approach. Before that, we need an additional result
regarding simulations, which allows us to view them as functions.

\begin{lemma}
\label{lem:funcSim}
    Let \(\dgraph_1 = (V_1, E_1, L_1)\) and \(\dgraph_2 = (V_2, E_2, L_2)\)
    be two \(\ourEL\) description graphs. Let \(Z : (\dgraph_1, v_1) \mapsto
    (\dgraph_2, v_2)\) be a simulation. Then, there exists a simulation \(Z'
    : (\dgraph_1, v_1) \mapsto (\dgraph_2, v_2)\) such that for every \(v
    \in V_1\), there is at most one \(w \in V_2\) such that \((v, w) \in
    Z'\).
\end{lemma}

\begin{proof}
    Assume that \(Z\) is a simulation from \((\dgraph_1, v_1)\) to
    \((\dgraph_2, v_2)\). If it satisfies the property that for every \(v
    \in V_1\), there is at most one \(w \in V_2\) such that \((v, w) \in
    Z\), we can take \(Z = Z'\).

    Let \(Z'\) be such that \(\{(v_1, v_2)\} \subseteq Z' \subseteq
    Z\). Also, suppose that it satisfies the following properties.

    \begin{itemize}
        \item If \((v, w) \in Z\), then exists \(w' \in V_2\) such
            that \((v, w') \in Z'\).
        \item If \((v, w) \in Z'\) and \((v, w') \in Z'\) then \(w
            = w'\).
    \end{itemize}

    A subset \(Z' \subseteq Z\) satisfying these two properties always exists: we
    just leave in \(Z'\) one pair \((v, w)\) for each \(v \in V_1\) such
    that \((v, w') \in Z\) for some \(w' \in V_2\).

    We will show that if \((v_1, v_2) \in Z'\), then \(Z'\) is a simulation
    from  \((\dgraph_1, v_1)\) to \((\dgraph_2, v_2)\).
    
    \begin{enumerate}
        \item \((v_1, v_2) \in Z'\) is assumed.
        \item If \((v, w) \in Z'\), then \((v, w) \in Z\),
            therefore \(L_1(v) \subseteq L_2(w)\) since \(Z\) is
            a simulation.
        \item If \((v, w) \in Z'\) and \((v, r, v') \in E_1\),
            then we know that \((w, r, w') \in E_2\) 
            and that \((v', w')
            \in Z\) for some \(w' \in V_2\). If \((v', w') \in Z'\)
            then (3) holds for \(Z'\). Otherwise, by construction of
            \(Z'\), there is a \(w'_2 \in V_2\) such that \((w', r, w'_2)
            \in E_2\) and \((v', w'_2) \in Z'\), which  proves (3) for
            \(Z'\) in this case.
    \end{enumerate}

    Therefore, \(Z' : (\dgraph_1, v_1) \mapsto (\dgraph_2, v_2)\) is
    a simulation such that for every \(v
    \in V_1\), there is at most one \(w \in V_2\) such that \((v, w) \in
    Z'\).
\end{proof}

Now we proceed to the claim's actual proof.

\dsim*

\begin{proof}
    Let \(\dgraph, \dgraph', v\) and \(v'\) as stated earlier and consider the
     unravellings  
    $$ 
     \dgraph_{d}^{v} = (V_d, E_d, L_d)
    \text{ and } 
    \dgraph_{k}^{v} = (V_{k}, E_{k}, L_{k})
    $$
    of \(\dgraph\).  Now, assume that there is
    a simulation \(Z_d : (\dgraph_d^{v}, v) \mapsto (\dgraph', v')\).  By
    Lemma~\ref{lem:funcSim}, we can assume w.l.o.g. that for
    each \(\dgpath{w} \in V_d\) there exists at most one \(u \in V'\) such that
    \((\dgpath{w}, u) \in Z_d\).  Therefore, we can define a function \(z\) such
    that \(z(\dgpath{w})\) is the only vertex in \(V'\) such that \((\dgpath{w},
    z(\dgpath{w})) \in Z_d\).

    If \(k \leq d\) then, 
    as \(\dgraph_k^v\) is a subtree
    of \(\dgraph_d^v\) (and thus, \(V_k \subseteq V_d\)), one can just take 
    \(Z_k = \{(\dgpath{w}, z(\dgpath{w})) \mid \dgpath{w} \in V_k\}\) as
    simulation.
%
    We now argue about the case where \(k > d\). Recall that the function
    \(\delta\) returns the vertex of a graph that 
    occurs at end of a path. 
    We show that in any
    path of length \(d\) in \(\dgraph_d^v\), there are two vertices
    \(\dgpath{w}_1\) and \(\dgpath{w}_2\) such that \(\delta(\dgpath{w}_1)
    = \delta(\dgpath{w}_2)\) and  \(z(\dgpath{w}_1) = z(\dgpath{w}_2)\).

    In what follows, we use the fact that unravellings are trees, and thus,
    for each vertex in an unravelling, there is exactly one path starting from the
    root to it.
     So we can refer to this path without ambiguity. Moreover, if
    \(\dgpath{w}\) is a vertex in an unravelling with root \(v\), then the path
    distance of \(\dgpath{w}\) is the length of the path from \(v\) to
    \(\dgpath{w}\).

    Now, let \(\dgpath{w} = w_0r_0\dots r_{n-1}w_n\) be a vertex in \(V_d\)
    and let \(\dgpath{w}_i = w_0r_0\dots r_{i-1}w_i\) for \(0 \leq i \leq n\) be
    the vertices in the path from \(v\) to \(\dgpath{w}\) (\(\dgpath{w}_0 = v\)
    and \(\dgpath{w}_n = \dgpath{w}\)).
    The path from \(v\) to \(\dgpath{w}\) in \(\dgraph_d^v\) determines a walk
    \(\dgpath{w}^*\) in \(\dgraph\) starting at \(v\) as follows:
    \[\dgpath{w}^* = \delta(\dgpath{w}_0)r_0\dots r_{n-1}\delta(\dgpath{w}_n).\]
    Due to the definition of \(\mvf\) there can be at most \(\mvf(\dgraph, v)\)
    distinct values of \(\delta\) for all vertices in the path from \(v\) to
    \(\dgpath{w}\), that is, \(|\{\delta(\dgpath{w}_i) \mid 0 \leq i \leq n\}|
    \leq \mvf(\dgraph, v)\). 

    As \(Z_d\) is a simulation, the path
    from \(v\) to \(\dgpath{w}\) also determines a walk in \(\dgraph'\)
    starting at \(v'\):
    \[\dgpath{w}' = z(\dgpath{w}_0)r_0\dots r_{n-1}z(\dgpath{w}_n).\]
    Again, due to the definition of \(\mvf\) there can be at most
    \(\mvf(\dgraph', v')\) distinct values of \(z\) for all vertices in the path
    from \(v\) to \(\dgpath{w}\), that is, \(|\{z(\dgpath{w}_i) \mid 0 \leq
    i \leq n\}| \leq \mvf(\dgraph, v)\).
    Therefore, there are at most \(\mvf(\dgraph,
    v) \cdot \mvf(\dgraph', v') = d\) distinct pairs \((\delta(\dgpath{w}'),
    z(\dgpath{w}'))\), where \(\dgpath{w}'\) is a vertex in the path from \(v\) to
    \(\dgpath{w}\), i.e.,
    \[|\{(\delta(\dgpath{w}_i), z(\dgpath{w}_i)) \mid 0 \leq i \leq n\}| \leq d.\]
    If a vertex \(\dgpath{w}\) has path 
    distance \(d\) from \(v\) in \(\dgraph_d^v\), then there are \(d + 1\)
    vertices in the path from \(v\) to \(\dgpath{w}\). As there are at most \(d\) 
    distinct pairs  \((\delta(\dgpath{w}'), z(\dgpath{w}'))\), where
    \(\dgpath{w}'\) is a vertex in this path, 
    and \(d+1\) vertices in the path from
    \(v\) to \(\dgpath{w}\), the pigeonhole principle implies that there will be
    two vertices \(\dgpath{w}_1, \dgpath{w}_2 \in V_d\) in the path from \(v\)
    to \(\dgpath{w}\) such that both \(z(\dgpath{w}_1) = z(\dgpath{w}_2)\) and
    \(\delta(\dgpath{w}_1) = \delta(\dgpath{w}_2)\).

Let \(\overline{V} \subseteq V_d\) be the set of all vertices
    such that there are no two distinct vertices
    \(\dgpath{w}_1\) and \(\dgpath{w}_2\) on the path from \(v\) to
    \(\dgpath{w}\) with 
    \(\delta(\dgpath{w}_1) = \delta(\dgpath{w}_2)\)
    and \(z(\dgpath{w}_1) = z(\dgpath{w}_2)\). Because of the previous
    argument, \(\overline{V}\) contains only vertices whose path distance
    from \(v\) is strictly less than \(d\). 

    Since \(\dgraph_d^v\) is a description tree with root \(v\),
    there is exactly one directed path from \(v\) to any given vertex
    \(\dgpath{w} \in V_d\). Hence, if \(\dgpath{w} \in \overline{V}\) then
    every vertex \(\dgpath{w}'\) on the path from \(v\) to \(\dgpath{w}\) in
    \(\dgraph_d^v\) is also in \(\overline{V}\). In other words, \(\overline{V}\)
    spans a subtree of \(\dgraph_d^v\).

    Now, let us consider the set \(\overline{V}^+\) composed by the direct successors
    of the leaves of the subtree determined by \(\overline{V}\), that is,
    \(\overline{V}^+ = \{w_0r_0\dots r_{n-1}w_n \in V_d \setminus \overline{V}
    \mid w_0r_0\dots r_{n-2}w_{n-1} \in \overline{V}\}\). Since each vertex in
    \(\overline{V}\) has path distance at most \(d - 1\) from \(v\), each vertex
    \(\overline{V}^+\) has path distance at most \(d\) from \(v\). Together with
    the fact that \(\overline{V}\) spans a subtree of \(\dgraph_d^v\), for each
    vertex \(\dgpath{w} \in V_d\) with path distance \(d\) from \(v\), there is
    exactly one vertex \(\dgpath{w}' \in \overline{V}^+\) in the path from \(v\)
    to \(\dgpath{w}\) (including the extremities).

    As we assume \(k > d\), we know that \(\dgraph_d^v\) is a subtree of \(\dgraph_k^v\),
    hence \(\overline{V}\) also spans a subtree of \(\dgraph_k^v\).  Therefore
    \(\overline{V} \cup \overline{V}^+ \in V_k\) and for every vertex
    \(\dgpath{w} \in V_k\) there is exactly one vertex \(\dgpath{w}'\) in
    \(\overline{V}^+\) in the path from \(v\) to \(\dgpath{w}\) in
    \(\dgraph_k^v\).  For each vertex \(\dgpath{w} \in V_k\), such
    \(\dgpath{w}'\) can be used to build a simulation from \(Z_d\) that includes
    \(\dgpath{w}\), as we will show next.

    For each vertex \(\dgpath{w} \in \overline{V} \cup \overline{V}^+\),
    there is exactly one vertex \(\dgpath{w}'\) in \(\overline{V}\) in the
    path from \(v\) to \(\dgpath{w}\) such that
    \(\delta(\dgpath{w}) = \delta(\dgpath{w}')\)
    and \(z(\dgpath{w}) = z(\dgpath{w}')\). Therefore, we
    can define a function \(s : \overline{V} \cup \overline{V}^+ \mapsto
    \overline{V}\) which retrieves such vertex for every \(\dgpath{w} \in
    \overline{V} \cup \overline{V}^+\).

    Now, we can use this function \(s\) to find an alternative path in \(V_d\)
    for each vertex in \(V_k\) when extending \(Z_d\) to the vertices in \(V_k
    \setminus V_d\). This notion is formalised by the function \(f : V_k \mapsto
    \overline{V}\) defined next, where \(\dgpath{w} = w_0r_0\dots
    w_{|\dgpath{w}|-1}r_{|\dgpath{w}|-1}w_{|\dgpath{w}|}\).
    \begin{align*}
         &f(\dgpath{w}) =\\
         &\qquad\begin{cases}
            s(\dgpath{w})  &\text{ if } \dgpath{w} \in \overline{V} \cup \overline{V}^+\\
            f(f(w_0r_0\dots w_{|\dgpath{w}|-1})r_{|\dgpath{w}|-1}w_{|\dgpath{w}|})
                &\text{ otherwise.}
        \end{cases}
    \end{align*}
    To clarify the rôle of \(f\) in this proof, consider a vertex
    \(\dgpath{w} = w_0r_0\dots r_{n-1}w_n \in V_k\) with \(n > d\). As before,
    let \(\dgpath{w}_i = w_0r_0\dots r_{i-1}w_i\) for \(0 \leq i \leq n\) be
    the vertices in the path from \(v\) to \(\dgpath{w}\). Since the path
    distance from \(v\) to
    \(\dgpath{w}\) is more than \(d\), we know
    that there is one \( 1 \leq m \leq n\) such that \(\dgpath{w}_m \in
    \overline{V}^+\). We also know that there is one \(0 \leq j < m\) such
    that \(s(\dgpath{w}_j) = \dgpath{w}_m\). When applying \(f\) to
    \(\dgpath{w}\), we obtain the following:
    \[f(\dgpath{w}) = f(\dots f \dots f(f(\dgpath{w}_m)r_mw_{m+1})
    \dots)r_{n-1}w_n).\]
    Since \(\dgpath{w}_m \in \overline{V}^+\), we have that \(f(\dgpath{w}_m)
    = \dgpath{w}_j\),
    which is closer to \(v\) than \(\dgpath{w}_m\). 
    As a consequence of \(s(\dgpath{w}_j) = \dgpath{w}_m\) and the definitions of
    unravelling and simulation, we know that \((\delta(\dgpath{w}_j), r_m,
    \delta(\dgpath{w}_{m+1})) \in E\) and \((z(\dgpath{w}_j), r_m,
    z(\dgpath{w}_{m+1})) \in E'\). Because the relation between vertices in
    \(\overline{V}^+\) and their image via the function \(s\) holds in each step of the
    recursion, we can add \((\dgpath{w}, z(f(\dgpath{w}))\) to \(Z_d\) for
    every vertex in \(V_k\) creating a new simulation.

    We use this observation to define the relation \(Z_k\) as:
    \[Z_k = \{(\dgpath{w}, z(f(\dgpath{w})) \mid \dgpath{w} \in V_k\}.\]


    Now we show that \(Z_k\) is a simulation from \((\dgraph^v_{k}, v)\) to
    \((\dgraph', v')\).

    \begin{enumerate}

        \item Since \(v \in \overline{V}\) and \(Z_d\) is a simulation satisfying
            the property of Lemma~\ref{lem:funcSim}, \((v, z(f(v))) = (v,
            z(v)) = (v, v')\).

        \item Since \(Z_d\) is a simulation and \(f(\dgpath{w}) \in V_d\):
				\begin{align*}
                L_k(\dgpath{w}) &= L(\delta(\dgpath{w}))
                = L(\delta(f(\dgpath{w})))\\ 
                 &= L_d(f(\dgpath{w})) \subseteq
                L'(z(f(\dgpath{w}))).
                \end{align*}

        \item Let \(\dgpath{w} \in V_k\) and assume that \((\dgpath{w},
            r, \dgpath{w}ry) \in E_k\). 

            If \(\dgpath{w}ry \in \overline{V} \cup \overline{V}^+\), then 
            \(\dgpath{w} \in \overline{V}\). Therefore, \((\dgpath{w},
            r, \dgpath{w}ry) \in E_d\). We also have: 

                \begin{align*}
                    (z(f(\dgpath{w})), r, z(f(\dgpath{w}ry))) &= (z(s(\dgpath{w})), r,
            z(s(\dgpath{w}ry)))\\
                &= (z(\dgpath{w}), r, z(\dgpath{w}ry)).
                \end{align*}

            Moreover, \((z(\dgpath{w}), r, z(\dgpath{w}ry)) \in
            E'\) because \(Z_d\) is a simulation. Finally, by construction, \((\dgpath{w}ry,
            z(\dgpath{w}ry)) \in Z_k\).

            Otherwise, if \(\dgpath{w}ry \not\in \overline{V} \cup
            \overline{V}^+\), we have that
            \(f(\dgpath{w}ry) = f(f(\dgpath{w})ry)\). Since \(f(\dgpath{w})
            \in \overline{V}\), \(f(\dgpath{w})ry \in \overline{V} \cup
            \overline{V^+}\) and consequently \(f(f(\dgpath{w})ry) = s(f(\dgpath{w})ry)\). By
            the definition of \(s\): \(z(s(f(\dgpath{w})ry))
            = z(f(\dgpath{w})ry) = z(f(\dgpath{w}ry))\). Since
            \(Z_d\) is a simulation and \(f(\dgpath{w}) \in V_d\), it
            holds that \((z(f(\dgpath{w})), r, z(f(\dgpath{w})ry))
            = (z(f(\dgpath{w})), r, z(f(\dgpath{w}ry))) \in E'\).
            Thus, \((\dgpath{w}ry, z(f(\dgpath{w})ry)) = (\dgpath{w}ry,
            z(f(\dgpath{w}ry)) \in Z_k\) which concludes the proof of (S3) for \(Z_k\).

    \end{enumerate}

    Therefore, \(Z_k\) is a simulation from
    \((\dgraph_{k}^v, v)\) to \((\dgraph', v')\), which proves the claim.
\end{proof}


Lemma~\ref{lem:boundedwalks} refers to walks in a product graph. To simplify its
proof we highlight a relationship between walks in the product graph and walks
in their factors via Proposition~\ref{prop:prod}.

\begin{proposition}
\label{prop:prod}
    Let \(\dgraph_1, \ldots,\dgraph_n\) be \(n\) description
    graphs, with \(\dgraph_i = (V_i, E_i, L_i)\) for $1\leq i\leq n$.
    It holds that, for each walk
    \(\dgpath{w}\) in \(\prod_{i = 1}^n \dgraph_i \) starting at \((v_1, \dots, v_n)\), there is
    a walk in \(\dgraph_i\) starting at \(v_i\) with the same length, for all $1\leq i\leq n$.
\end{proposition}

\begin{proof}
    Let \(\dgpath{w}\) be a walk in \(\prod_{i = 1}^n \dgraph_i\) starting in
    \((v_1, \dots, v_n)\) with length \(m\) as follows:
    \[\dgpath{w} = (w_{1,0}, \dots, w_{n,0})r_0 \dots r_{m-1} (w_{1,{m-1}},
    \dots, w_{n,m})).\] 
    The walk \(\dgpath{w}_i = w_{i,0}r_0 \dots r_{m-1} w_{i,{m}}\) is a walk in
    \(\dgraph_i\) because \(w_{i,j} \in V_i\) for \(0 \leq j < m\)
    and \((w_{i,j}, r_{j}, w_{i,j+1}) \in E\) due to the
    definition of product. As \(w_{i,0} = v_i\), by construction of
    \(\dgpath{w}\), \(\dgpath{w}_i\) 
    starts at \(v_i\). Additionally, by
    construction, \(\dgpath{w}_i\) has length \(m\), which concludes the proof.
\end{proof}

We use Proposition~\ref{prop:prod} 
in Lemma~\ref{lem:boundedwalks} below.

\boundedwalks*

\begin{proof}
    Let \(X = \{x_1, \dots, x_n\} \subseteq \Delta^\I\) and let
    \begin{align*}
        &X_{lim} = \{x \in X \mid \exists m \in \mathbb{N} : \text{every walk} \\
        &\qquad \text{starting from } x \text{ in } \dgraph(\I) \text{ has length} \leq m\}.
    \end{align*}

    Assume \(X_{lim} \neq \emptyset\) and let \(x' \in X_{lim}\) be such that
    \[\mvf(\dgraph(\I), x') = \min_{x \in X_{lim}} \mvf(\dgraph(\I),
    x).\] Since \(x' \in X_{lim}\), every walk in \(\dgraph(\I)\) starting at
    \(x'\) has length bounded by \(\mvf(\dgraph(\I), x') - 1\). Due to the
    definition of product of description graphs (recall how the edges are built), 
    this limitation extends to every walk in
    \(\prod_{i=1}^n \dgraph(\I)\) starting at \((x_1, \dots, x_n)\): they
    have length at most \(\min_{x \in X_{lim}} \mvf(\dgraph(\I), x)
    - 1\). If there was a longer walk, there would be also a walk in 
    in \(\dgraph(\I)\) starting at \(x'\) with the same length due to
    Proposition~\ref{prop:prod}.
\end{proof}


In the following, we prove that our adaptable role depth yields an upper bound of the actual
fixpoint for an MMSC.

\mmscK*

\begin{proof}

    Let \(X = \{x_1, \dots, x_n\} \subseteq \Delta^\I\) and

    \begin{align*}
        &X_{lim} = \{x \in X \mid \exists m \in \mathbb{N} : \text{every walk} \\
        &\qquad \text{starting from } x \text{ in } \dgraph(\I) \text{ has length} \leq m\}.
    \end{align*}

    If \(k \leq \newd{X}{\I}\) the lemma holds trivially. For \(k
    > \newd{X}{\I}\) we divide the proof in two cases. First, if \(X_{lim} \neq
    \emptyset\) then as stated in Lemma~\ref{lem:boundedwalks}, every walk in
    \(\prod_{i=1}^n \dgraph(\I)\) starting at \((x_1, \dots, x_n)\) has length
    at most \(\mvf\left(\prod_{i=1}^n \dgraph(\I), (x_1, \dots, x_n)\right)
    - 1 = \newd{X}{\I}\).

    In other words, even when \(k > \newd{X}{\I}\), we have: \(\prod_{i=1}^n
    \dgraph(\I)_k^x = \prod_{i=1}^n \dgraph(\I)_\newd{X}{\I}^x\), and
    therefore, we can apply Lemma~\ref{lem:memberSimu1} to conclude that:

    \[\mmsc{X}{\I}[\newd{X}{\I}]^\I \subseteq \mmsc{X}{\I}[k]^\I.\]

    Otherwise, if \(X_{lim} = \emptyset\), we can use the fact that 
    \(\mmvf(\dgraph) \geq \mvf(\dgraph, x')\; \forall x' \in \Delta_\I\) to
    obtain:

    \[ \newd{X}{\I} \geq \mvf\left(\prod_{i = 1}^n\dgraph(\I), (x_1, \dots,
    x_n)\right) \cdot \mvf(\dgraph(\I), x').\]

    Hence, if \(X_{lim} = \emptyset\), the lemma is a direct consequence
    of Definition~\ref{def:newd} and Lemma~\ref{lem:unravelExt}.
\end{proof}


To prove that \(\base(\I)\) defined in Theorem~\ref{thm:b2} is a base,
we  first  
 recall a result related to the notion of MMSC.

\begin{lemma}{\cite{Borchmann2016}}
\label{lem:r5}
    Let \(\I = (\Delta^\I, \cdot^\I)\) be a finite \(\ourEL\) interpretation.
 For all \(X \subseteq \Delta^\I\) and \(k \in \mathbb{N}\), it holds that 
%
%
    $\emptyset \models \mmsc{\mmsc{X}{\I}[k]^\I}{\I}[k] \equiv \mmsc{X}{\I}[k].$
\end{lemma}

We will also need a property regarding the construction of concept expressions
with MMSCs.

\begin{lemma}[Adaptation of Proposition A.1 from \cite{Borchmann2016}]
\label{lem:monotonicity}
    For all \(\ourEL\) concept expressions \(C, D\) over \(\NC \cup \NR\)
    and all \(r \in \NR\) it holds that: 
    \[(\mmsc{C^\I}{\I} \sqcap D)^\I = (C \sqcap D)^\I,\]
    \[(\exists r.(\mmsc{C^\I}{\I}))^\I = (\exists r.C)^\I.\]
\end{lemma}


Then, we define for each concept expression \(C\) and interpretation \(\I\)
a specific concept in \(\Lambda_\I\) which is called the lower approximation
of \(C\) in \(\I\). We recall that, for $ X\subseteq \Delta^\I $, we write \(\mmsc{X}{\I}\)
as a shorthand for \(\mmsc{X}{\I}[\newd{X}{\I}]\). 

\begin{definition}[Lower Approximation (adapted from Definition 5.4 in \cite{Distel2011})]
\label{def:approx}
    Let \(C\) be an \(\ourEL\) concept expression and \(\I = (\Delta^\I,
    \cdot^\I)\) a model. Also let \(\NC \cup \NR\) be a finite signature and 
    \(\ourEL(\NC, \NR)\) the set of all \(\ourEL\) concept expressions over  
    \(\NC \cup \NR\).
    Then, there are concept names \(U \subseteq \NC\) and pairs \(\Pi \subseteq \NR
    \times \ourEL(\NC, \NR)\) such that:

    \[C = \bigsqcap U \sqcap \bigsqcap_{(r, E) \in \Pi} \exists
    r.E\]
   
    We define the lower approximation of \(C\) in \(\I\)  as:
    \begin{align*}
        &\app(C, \I) =\\
        &\qquad\begin{cases}
            \bigsqcap U \sqcap \bigsqcap_{(r, E) \in \Pi} \exists
            r.\mmsc{E^\I}{\I} & \text{ if } C \neq \bot,\\
            \bot & \text{otherwise.}
        \end{cases}
    \end{align*}
    
\end{definition}

Concept expressions built according to Definition~\ref{def:approx} are always
elements of \(\Lambda_\I\) because they are a conjunction of elements in
\(M_\I\) (Definition~\ref{def:setofconceptexpressions:MI}).
Next, with a straightforward, but nevertheless important, adaptation of the
Lemma 5.8 from \cite{Distel2011} we prove that the lower approximation of
a concept and the concept itself have the same extension.

\begin{lemma}
\label{lem:appSub}
    Let \(C\) be an \(\ourEL\) concept expression and \(\I = (\Delta^\I,
    \cdot^\I)\) a model. It holds that
    \[\mmsc{C^\I}{\I}^\I = \app(C, \I)^\I = C^\I.\]
\end{lemma}

\begin{proof}
    If \(C = \bot\) then \(\mmsc{C^\I}{\I}^\I = \app(C, \I)^\I
    = \emptyset\).
    Otherwise, there are concept names \(U \subseteq \NC\) and pairs
    \(\Pi \in \NR \times \ourEL(\NC, \NR)\) such that

    \[C = \bigsqcap U \sqcap \bigsqcap_{(r, E) \in \Pi} \exists
    r.E\]



    Using Lemma~\ref{lem:monotonicity} we obtain:

    \begin{align*}
        C^\I &= (\bigsqcap U \sqcap \bigsqcap_{(r, E) \in \Pi}
        \exists r.E)^\I\\
        {} &= (\bigsqcap U \sqcap \bigsqcap_{(r, E) \in \Pi} \exists
        r.\mmsc{E^\I}{\I})^\I\\
        {} &= (\app(C, \I))^\I
    \end{align*}

    Finally, we can apply Lemma~\ref{lem:r6p} obtaining
    \(\mmsc{C^\I}{\I}^\I = \app(C, \I)^\I\).
\end{proof}

Using these results, we can conclude that for each MMSC there is a concept expression in
\(\Lambda_\I\) with the same extension in \(\I\). 
With this observation we can 
we can proceed to Theorem~\ref{thm:b2}'s proof.




\thmbase*

\begin{proof}
    As \(\Lambda_\I\) is finite, so it is \(\base(\I)\). The concept inclusions
    are clearly sound and the soundness of the equivalences is due to 
    Lemma~\ref{lem:r6p}.


    Let \(\J = (\Delta^\J, \cdot^\J)\) be an arbitrary interpretation such that
    \(\J \models \base(\I)\).
    For completeness, we prove that for any \(\ourEL\) concept expression \(C\),
    \(\J \models C \equiv \mmsc{C^\I}{\I}\).
    We prove this claim by induction of the structure of \(C\).
	
	\paragraph{Base case:} If \(C = \bot\) or \(C = A\) where $A\in\NC$,
	then $C\in\Lambda_\I$, by definition of $\Lambda_\I$. 
	 %
    Then, by definition of $\base(\I)$, we have that \( C \equiv
    \mmsc{C^\I}{\I}
	\in \base(\I)\). 
	
	\paragraph{Step case (\(\sqcap\)):} Suppose \(C = E \sqcap F\) and 
	the claim holds for $E$ and $F$. By the 
	inductive hypothesis, 
	\(\base(\I) \models E \equiv
    \mmsc{E^\I}{\I}\) and \(\base(\I) \models F \equiv
    \mmsc{F^\I}{\I}\). 
	Hence, for all interpretations $ \J $ such that $ \J \models \base(\I) $,
	we have that \(E^\J
    = \mmsc{E^\I}{\I}^\J\) and \(F^\J
    = \mmsc{F^\I}{\I}^\J\).
%
	%
    By Lemma~\ref{lem:appSub}, 
    there are 
	\(\overline{E},\overline{F}\in\Lambda_\I\)
    such that \(\mmsc{E^\I}{\I}^\I = \overline{E}^\I\) and 
    \(\mmsc{F^\I}{\I}^\I = \overline{F}^\I\). 
    Moreover, by Lemma~\ref{lem:r6p}, 
    \(\mmsc{E^\I}{\I}^\I = E^\I\) and
    \(\mmsc{F^\I}{\I}^\I = F^\I\).
	Therefore 
	\((\overline{E} \sqcap \overline{F})^\I = 
	\overline{E}^\I \cap \overline{F}^\I =
	E^\I \cap F^\I = (E\sqcap F)^\I \).
	
	As
	\(\overline{E} \sqcap \overline{F} \in \Lambda_\I\)
    (up to logical equivalence),   \(\overline{E} \sqcap \overline{F} \equiv
    \mmsc{(\overline{E} \sqcap \overline{F})^\I}{\I} \in \base(\I)\) (again up to logical equivalence).
	Since \(\J\) is a model of \(\base(\I)\), by
	Lemma~\ref{lem:monotonicity}:
	\begin{align*}
	\left(\overline{E} \sqcap \overline{F}\right)^\J &=
	\mmsc{(\overline{E} \sqcap \overline{F})^\I}{\I}^\J\\
    {} &= \mmsc{(E \sqcap F)^\I}{\I}^\J\\
    {} &= \mmsc{C^\I}{\I}^\J.
	\end{align*}
	
	To prove that $C^\J
    = \mmsc{C^\I}{\I}^\J$, in the following, we write $\overline{C}$ as a shorthand for $\overline{E} \sqcap \overline{F}$
	and
	show that $  \overline{C}^\J = C^\J $.
    Since $ \overline{E}  \in \Lambda_\I$,
    we have that $ \base(\I) \models \overline{E} \equiv
    \mmsc{\overline{E}^\I}{\I} $.
    Moreover, as \(\mmsc{E^\I}{\I}^\I = \overline{E}^\I\),
    we have that
    \[ \base(\I) \models \overline{E} \equiv
        \mmsc{\mmsc{E^\I}{\I}^\I}
    {\I}. \]
    By Lemma~\ref{lem:r5}, it follows that 
    \[\emptyset \models \mmsc{\mmsc{E^\I}{\I}^\I}
        {\I} \equiv \mmsc{E^\I}{\I}.\]
    Therefore,
    $ \base(\I) \models \overline{E} \equiv \mmsc{E^\I}{\I}$
    and as 
    $ \base(\I) \models E \equiv \mmsc{E^\I}{\I}$,
    then $ \base(\I) \models E \equiv \overline{E}$.
    Similarly we obtain 
    $ \base(\I) \models \overline{F} \equiv F $
    and that $ \base(\I) \models \overline{C} \equiv C $.
    As $ \J $ was  an arbitrarily chosen model of $ \base(\I) $, 
    we  conclude that 
    $ \base(\I) \models \overline{C} \equiv  
    \mmsc{C^\I}{\I} $ and  
    $\base(\I) \models \overline{C} \equiv C$.

	\paragraph{Step case (\(\exists\)):} In this case, \(C = \exists r.E\)
	for some \(r \in \NR\) and \(\ourEL\) concept expression \(E\). Let
	\(\J\) be an interpretation such that \(\J \models \base(\I)\). We
	know that:
	
	\begin{align*}
	x \in C^\J &\iff x \in (\exists r.E)^\J\\
	{} &\iff \exists y \in E^\J : (x, y) \in r^\J.
	\end{align*}
	By our induction hypothesis, \(\base(\I) \models E \equiv
    \mmsc{E^\I}{\I} \), hence:
	\begin{align*}
        x \in C^\J &\iff \exists y \in \mmsc{E^\I}{\I}^\J : (x, y) \in r^\J\\
        {} &\iff x \in (\exists r.\mmsc{E^\I}{\I})^\J.
	\end{align*}
	
    In short, we proved that $ C^\J= (\exists r.\mmsc{E^\I}{\I})^\J $.
    Next, as \(\exists r.\mmsc{E^\I}{\I} \in
	M_\I\), we know that 
	\begin{align*}
    &	\exists r.\mmsc{E^\I}{\I}  \equiv  \\
	& \mmsc{\exists r.\mmsc{E^\I}{\I}^\I}{\I}
	\in \base(\I)
	\end{align*}
	%
    With Lemma~\ref{lem:monotonicity} we obtain:
	
	\begin{align*}
        (\exists r.\mmsc{E^\I}{\I})^\J &=
        (\mmsc{\exists 	r.\mmsc{E^\I}{\I}^\I}{\I})^\J\\ 
        {} &= (\mmsc{(\exists r.E)^\I}{\I})^\J\\
        {} &= (\mmsc{C^\I}{\I})^\J.
	\end{align*}
	
    Thus, \(C^\J = (\mmsc{C^\I}{\I})^\J\). Since \(\J\) was
	chosen arbitrarily, we can conclude that \(\base(\I) \models C \equiv
    \mmsc{C^\I}{\I}\). 

    Now, we prove that if \(\I \models C \sqsubseteq D\), then
    \(\base(\I) \models \mmsc{C^\I}{\I} \sqsubseteq
    \mmsc{D^\I}{\I}\). Let \(\J\) be a model of
    \(\base(\I)\). We know from Lemmas~\ref{lem:r6p}
    and~\ref{lem:appSub}   that there are \(U, V \subseteq M_\I\) such that
    \(C^\I = (\bigsqcap U)^\I\) and \(D^\I = (\bigsqcap V)^\I\). From the
    definition of \(\base(\I)\), we obtain \( \mmsc{(\bigsqcap
    U)^\I}{\I} \sqsubseteq \mmsc{(\bigsqcap
    V)^\I}{\I} \in \base(\I)\). Therefore:

    \[
        \J \models \mmsc{(\bigsqcap U)^\I}{\I}
    \sqsubseteq \mmsc{(\bigsqcap V)^\I}{\I}
    \]

    Replacing \((\bigsqcap U)^\I\) with \(C^\I\) and \((\bigsqcap V)^\I\)
    with \(D^\I\) yields:

    \[
        \J \models \mmsc{C^\I}{\I} \sqsubseteq
        \mmsc{D^\I}{\I}.
    \]

    Therefore, using the fact that \(\J \models C \equiv \mmsc{C^\I}{\I}\) for
    every \(\ourEL\) concept expression \(C\) (proved earlier)
    we can conclude that \(\J \models C \sqsubseteq D\).

    Since all the required concept inclusions hold in an arbitrary model of
    \(\base(\I)\), whenever they hold in \(\I\) we have that
    \(\base(\I)\) is also complete for the \(\ourEL\) CIs.
\end{proof}

\section{Proofs for Section~\ref{sec:mvf}}

In the following we present the proofs related to the  computation of the MVF function.
In particular, we provide proofs for
the relationship between 
the condensation graph and the MVF function
(Lemma~\ref{lem:mvfCondense}), and the correctness of Algorithm~\ref{alg:mvf}
(Lemma~\ref{lem:maxWeight}). 

\mvfCondense*

\begin{proof}
    First we prove that every path \(\dgpath{w}^* = V_1, \dots, V_m\) in
    \(\dgraph^*\) starting
    at \(\scc(\dgraph, v)\) induces a walk in \(\dgraph\) starting at \(v\) with
    \(\lend(\dgpath{w}) = \weight(\dgpath{w}^*)\).  Let \(v_1 = v\). For each
    \(1 \leq i < m\), the induced
    walk must: visit \(v_i\), then pass
    through all vertices in \(V_i\) (repeating vertices whenever needed), then
    visit a vertex \(u_i \in V_i\) such that there is an edge \((u_i, r,
    v_{i+1}) \in E\) with \(v_{i+1} \in V_{i + 1}\) (this is possible due to the
    definitions of SCCs and condensation). When the walk reaches a vertex
    \(u_{m-1}\), it must visit \(v_m\) and 
    pass through every vertex in \(V_m\) before stopping.
    Such walk visits every vertex in \(\bigcup_{i=1}^m V_i\), thus
    \(\lend(\dgpath{w}) = \weight(\dgpath{w}^*)\).

    Now let \(\dgpath{w}\) be a walk in \(\dgraph\) starting at \(v\) which is
    induced (as explained earlier) by a path \(\dgpath{w}^*\) in \(\dgraph^*\) starting at
    \(\scc(\dgraph, v)\) with maximum weight. Assume that there is a walk
    \(\overline{\dgpath{w}}\) in \(\dgraph\) starting at \(v\)  such that
    \(\lend(\overline{\dgpath{w}}) > \lend(\dgpath{w})\). Due to the definitions
    of SCC and condensation we know that there is a path
    \(\overline{\dgpath{w}}^*\) in \(\dgraph^*\) starting at \(\scc(\dgraph,
    v)\) such that \(\lend{\overline{\dgpath{w}}} \leq
    \weight(\overline{\dgpath{w}}^*)\). However, this would imply that:
    \(\weight(\dgpath{w}^*) = \lend(\dgpath{w}) < \lend{\overline{\dgpath{w}}}
    \leq \weight(\overline{\dgpath{w}}^*)\), which is a contradiction since we
    assume that \(\dgpath{w}^*\) has maximal weight. Therefore, no walk in
    \(\dgraph\) starting at \(v\) can visit more vertices than
    \(\weight(\dgpath{w}^*)\).


    Since we have shown that for every path \(\dgpath{w}^*\) in
    \(\dgraph^*\) starting at \(\scc(\dgraph, v)\), there is a walk
    \(\dgpath{w}\) in \(\dgraph\) starting at \(v\), with \(\lend(\dgpath{w})
    = \weight(\dgpath{w}^*)\), we can conclude that the statement of this lemma
    holds.
\end{proof}

\lemMaxWeight*

\begin{proof}
	Let \(\dgraph^* = (V^*, E^*)\) be the condensation of \(\dgraph\). If 
	\(W' \in V^*\) is unreachable from \(\scc(\dgraph, v)\) then \({wgt}[V']\) will
    remain null as it will never be visited. Otherwise, \(W'\) will be visited
    in some call of \({\sf maxWeight}\). If it has no successors,
    the loop in
    Line~\ref{line:successors} will not do anything, and thus \({wgt}[W']
    = |W'|\) as expected.  Instead, 	
	if \(\scc(\dgraph, v)\) has successors, then the maximum
	weight of a path starting at \(\scc(\dgraph, v)\) in \(\dgraph^*\) is given by \(|\scc(\dgraph, v)|\)
	plus the maximum value computed among its successors. This equation holds
	because \(\dgraph^*\) is a DAG.  
    Since, the loop in
    Line~\ref{line:successors} forces the maximum weights of the successors of
    \(W'\) to be calculated first, the value returned in
    Line~\ref{line:recursiveReturn} is correct.	
\end{proof}

\fi

\end{document}